\pdfoutput=1
\documentclass[letterpaper,11pt]{article}

\usepackage{style}
\usepackage{shortcuts}

\title{$\ell_2/\ell_2$ Sparse Recovery via Weighted Hypergraph Peeling}
\author{%
    Nick Fischer\thanks{Max Planck Institute for Informatics, \texttt{nfischer@mpi-inf.mpg.de}. Parts of this work were done while the author was affiliated with INSAIT, Sofia University ``St.\ Kliment Ohridski''. This work is partially funded by the Ministry of Education and Science of Bulgaria (support for INSAIT, part of the Bulgarian National Roadmap for Research Infrastructure).}
\and%
    Vasileios Nakos\thanks{National and Kapodistrian University of Athens and Archimedes / Athena RC, \texttt{vasilisnak@di.uoa.gr}.}}
\date{}

\begin{document}

\maketitle

\begin{abstract}
\noindent
We demonstrate that the best $k$-sparse approximation of a length-$n$ vector can be recovered within a $(1+\epsilon)$-factor approximation in $O( (k/\epsilon) \log n)$ time using a non-adaptive linear sketch with $O( (k/\epsilon ) \log n)$ rows and $O(\log n)$ column sparsity. This improves the running time of the fastest-known sketch [Nakos, Song; STOC '19] by a factor of $\log n$, and is \emph{optimal} for a wide range of parameters.

Our algorithm is simple and likely to be practical, with the analysis built on a new technique we call \emph{weighted hypergraph peeling}. Our method naturally extends known hypergraph peeling processes (as in the analysis of Invertible Bloom Filters) to a setting where edges and nodes have (possibly correlated) weights.
\end{abstract}

\setcounter{page}{0}
\thispagestyle{empty}
\clearpage

\section{Introduction} \label{sec:intro}
The sparse recovery problem poses a deceptively simple yet profound question: Is it possible to approximately reconstruct the best $k$-term approximation of a vector $x\in\mathbb{R}^n$ using only a few linear queries, and how efficiently can this be achieved? This question, in its various forms and under different constraints on the queries and computational resources, is a central challenge in fields like signal processing~\cite{CandesT05, CandesT06, CandesRT06}, data compression~\cite{LiXW13a}, streaming algorithms under the name of the \emph{heavy hitters} problem~\cite{CharikarCF04, LarsenNNT19}, as well as applied mathematics, usually in the form of \emph{combinatorial group testing}~\cite{DuH99}. By focusing on the recovery of sparse vectors---where most of the energy is concentrated in a small number of coordinates---research in this area has unlocked new possibilities across a wide range of applications, from MRI imaging~\cite{LustigDSP08, JaspanFL15} to real-time data analysis~\cite{LiXW13a}. The exploration of this problem has fundamentally transformed how we handle vast amounts of data, enabling more efficient data storage, transmission, and processing in diverse industries~\cite{LustigDSP08,JaspanFL15,EldarK12, HermanS09}.

In this paper, we focus on one of the most well-studied forms of the sparse recovery problem~\cite{CharikarCF04, CormodeM06, GilbertI10, PriceW11, Price11, GilbertLPS12,  NakosS19}. The goal is to design a (typically \emph{randomized}) matrix $A\in\mathbb{R}^{m\times n}$ that satisfies the following criteria: (i) given the measurement $Ax$ for~\makebox{$x \in \mathbb{R}^n$}, we can recover a vector~$x'$ that retains most of the energy from the top $k$ coordinates of~$x$, (ii)~the number of rows (also referred to as measurements), $m$, is minimized, (iii)~the encoding time, corresponding to the column sparsity of $A$, is as small as possible, (iv)~the recovery procedure is fast, ideally even \emph{sublinear-time} in $n$. 

The celebrated work of Candes, Tao, Donoho and others~\cite{CandesT05, CandesT06, CandesRT06} has demonstrated that $m = O(k \log (n/k))$ rows are sufficient for approximate recovery \emph{for all vectors} in time that is polynomial in $n$ under a guarantee usually referred to as~$\ell_2/\ell_1$. Alternatively, it has been shown how to recover such an $x'$ in time $O(n \log (n/k))$ if the notion of approximation is relaxed further to an~$\ell_1/\ell_1$ guarantee~\cite{IndykR08}. Both guarantees have been accomplished with sublinear-time algorithms, albeit with a polylogarithmic measurement overhead and polynomial in $m$ recovery time for the former~\cite{GilbertSTV07}, and a slight sub-optimality for the latter due to a multiplicative factor of $(\log_k n)^\gamma$ factor for some constant $\gamma <1$ in the number of measurements~\cite{GilbertLPS17}.

The recovery guarantee we focus on in this work is the so-called $\ell_2/\ell_2$ guarantee, which is strictly stronger than the previously mentioned $\ell_2/\ell_1$ and $\ell_1/\ell_1$ guarantees and can only be achieved through the use of randomness~\cite{Cohen2008}, i.e., no deterministic matrix can recover \emph{all vectors}. Specifically, the $\ell_2/\ell_2$ guarantee we are shooting for demands that the recovered vector $x'$ satisfies
\begin{equation*}
    \norm{x-x'}_2 \leq (1+\epsilon) \cdot\! \min_{z: \norm{z}_0 \leq k} \norm{x-z}_2
\end{equation*}
This guarantee is arguably the most well-studied one in the context of sublinear-time sparse recovery, appearing in various forms across the literature~\cite{CormodeM06, PriceW11, Price11, GilbertLPS12, HassaniehIKP12a, HassaniehIKP12b, Price13, Kapralov16, Kapralov17, IndykK14, IndykKP14, CevherKSZ17, NakosS19, CheraghchiN20}. Any (randomized) matrix which enables $\ell_2/\ell_2$ sparse recovery with constant probability must have at least $m= \Omega( (k/\epsilon) \log (n/k))$ rows~\cite{PriceW11}, and this lower bound is matched by upper bounds~\cite{GilbertLPS12, NakosS19}. Specifically, the state of the art is the work of Nakos and Song~\cite{NakosS19} which achieves the aforementioned optimal bound on~$m$ with column sparsity $O(\log n)$ and recovery time $O(m \log n)$, improving upon the previous work of Gilbert, Li, Porat and Strauss~\cite{GilbertLPS12} by polylogarithmic factors in both recovery time and column sparsity. While the number of rows is best-possible, however, there is no lower bound for the query time $O(m \log n)$.

Our contribution is a concise algorithm for the $\ell_2/\ell_2$ sparse recovery problem achieving $m = O((k/\epsilon) \log n)$ rows, column sparsity $O(\log n)$, and most importantly recovery time $O(m)$:

\begin{restatable}[Fast $\ell_2/\ell_2$ Sparse Recovery]{theorem}{thmmain} \label{thm:main}
Let $k \leq n$ be integers and let $0 < \epsilon < 1$. There is a randomized construction of a matrix $A\in\mathbb{R}^{m\times n}$ with $m=O((k/\epsilon) \log n)$ rows and $O( \log n)$ non-zero entries per column, such that given $Ax$ we can recover in time $O(m)$ a vector $x'$ satisfying 
\[ \|x- x'\|_2 \leq (1+\epsilon) \|x_{-k}\|_2,\]
with arbitrarily large constant probability. Here, $x_{-k}$ is the vector obtained from $x$ by zeroing out the largest~$k$ entries in absolute value.
\end{restatable}

Notably, our algorithm is \emph{optimal} whenever $\log n = \Theta(\log (n/k))$, i.e., whenever $k \leq n^{1-\zeta}$ for some arbitrarily small constant $\zeta > 0$. Running time-wise, it is the first algorithm to achieve time $O(m)$ for the problem. Generally, an algorithm with optimal $m$ and optimal recovery time $O(m)$ is the holy grail of sublinear-time sparse recovery, and to our knowledge our algorithm is the \emph{first} such algorithm for any robust\footnote{I.e., where the vector $x$ is not \emph{exactly} $k$-sparse.} sparse recovery problem. The recovery algorithm itself is very simple (we give a description below) and all the complexity is pushed to the analysis, using a new technique we call \emph{weighted hypergraph peeling}.

Additionally, we believe this algorithm has significant practical potential, as it only requires a black-box $1$-sparse recovery routine and a standard Count-Sketch data structure~\cite{CharikarCF04, MintonP14} to query the values of $x$. While the two main previous sublinear-time approaches~\cite{GilbertLPS12, NakosS19} require fine-tuning of parameters, our approach is distinct and provides deeper insights into the structure of $\ell_2/\ell_2$ sparse recovery, both from a technical and a conceptual perspective.

\paragraph{Algorithm Description.}
As promised, we give a nearly complete description of the recovery algorithm in the introduction. Each coordinate of the vector $x$ is hashed to some constant number~$h \geq 2$ out of $O(k/\epsilon)$ buckets, where each bucket is augmented with an $O(\log n)$-row matrix that allows fast $1$-sparse recovery with high probability. This yields $O((k/ \epsilon) \log n)$ rows in total. On the side we use a standard Count-Sketch with $O((k/ \epsilon) \log n)$ rows which allows estimating each individual $x_i$ upon demand~\cite{CharikarCF04, MintonP14}. For the sake of exposition, here we assume that we know $\|x_{-k}\|_2$.

\bigskip
\begin{adjustbox}{frame={\fboxrule} 8pt,minipage={\textwidth-2\fboxrule-16pt}}
\begin{algorithmic}[1]
    \State Initialize $R \gets \emptyset$ and let $Q$ be the set of all buckets
    \While{$Q \neq \emptyset$}
        \State Take an arbitrary bucket $v \in Q$ and remove $v$ from $Q$
        \State Run the $1$-sparse recovery routine on bucket $v$ to obtain a candidate $i$
        \State Estimate $x_i$ using Count-Sketch to obtain $\hat{x}_i$
        \If{$|\hat x_i|^2 \geq \frac{\epsilon}{2k} \cdot \norm{x_{-k}}_2^2$ \AND{} $i \not\in R$} 
            \State Update $R \gets R  \cup \set{i}$ and (re-)insert to $Q$ all buckets that $i$ participates in 
            \State Subtract the effect of $\hat{x}_i$ from all the buckets $i$ participates in
        \EndIf
    \EndWhile
    \State Sort $R = \set{i_1, \dots, i_{|R|}}$ such that $|\hat x_{i_1}| \geq \dots \geq |\hat x_{i_{|R|}}|$
    \State Let $S = \set{i_1, \dots, i_{\min(3k, |R|)}}$
    \State\Return $x' \gets \hat x_S$
\end{algorithmic}
\end{adjustbox}
\bigskip
\section{Technical Overview} \label{sec:overview}
In this section we give a high-level overview of our new sketch. It relies on hashing at its main building component, similarly to \cite{GilbertLPS12, NakosS19} and in all the heavy hitters literature. However, our specific approach differs quite substantially from previous work. We briefly review the two known approaches.

\paragraph{Known Approaches: Iterative and Hierarchical.}
In~\cite{GilbertLPS12} the vector $x$ is hashed to $O(k/\epsilon)$ buckets, augmented with an appropriate error-correcting code in such a way that from all buckets that contain \emph{only one} heavy hitter that heavy hitter is recoverable. However, we only expect a constant fraction, say $k/2$, of the heavy hitters to be isolated in their respective buckets. So letting~$x'$ denote the resulting recovered vector, there are still up to $k/2$ heavy hitters remaining in the \emph{residual vector} $x - x'$. To recover these, one can proceed in an iterative fashion, i.e., we essentially repeat the same hashing process where the number of buckets can be reduced by a constant factor. This process can be set up in a such way that eventually, after $O(\log k)$ rounds, it yields a vector satisfying the $\ell_2/\ell_2$ guarantee. Due to the iterative procedure, one has to pay at least one $\log n$ factor in the column sparsity and decoding time. For details, we refer to~\cite{GilbertLPS12}. 

In contrast, the improved algorithm in~\cite{NakosS19} builds a $\log (k/\epsilon)/ \log \log (k/\epsilon)$-degree interval tree over $[n]$ and traverses that tree from top to bottom in order to find where the heavy hitters are, recognizing \emph{heavy} intervals per level. At each level elements belonging to the same interval are hashed to the same bucket and combined with \emph{gaussians}, instead of the usual random signs. The choice of gaussians is crucial in order to reduce the likelihood of cancellation, and using the anti-concentration of the gaussian distribution it is shown that \emph{most} energy from heavy hitters is retained if one sets up that process correctly. The branching factor of this process yields one additional $\log $ factor in the decoding time, and changing the branching factor crucially affects correctness. For more details, we refer to~\cite{NakosS19}.

\paragraph{New Approach: One-Shot.}
In a nutshell, as we are aiming for the best-possible $O((k / \epsilon) \log n)$ recovery time, our goal is to perform \emph{one-shot} recovery rather than iterative recovery (as in~\cite{GilbertLPS12}) or hierarchical recovery (as in~\cite{NakosS19}). That is, we will throughout consider a single hashing scheme.

Before we get into the details of our new approach, let us first fix some definitions. We call $i$ a \emph{heavy hitter} (or just \emph{heavy}) if $|x_i|^2 \geq \frac{\epsilon}{k}  \cdot \norm{x_{-k}}_2$, and otherwise we call $i$ \emph{light}. Let $T \subseteq [n]$ denote the set of the largest $k$ entries in $x$. Following~\cite{NakosS19}, instead of computing $x'$ directly we resort to a slight simplification of the original task: It suffices to recover a subset of heavy hitters $R$ with the property that we do not miss too much energy from $T$, i.e., such that
\begin{equation*}
    \norm{x_{T \setminus R}}_2^2 \leq O(\epsilon) \norm{x_{-k}}_2^2.
\end{equation*}
One can then easily refine $R$ to derive the desired approximation $x'$ of $x$. For the rest of this overview we elaborate on how to compute such a set $R$.

\paragraph{Inspiration: Hypergraph Peeling.}
Our algorithm is inspired by hypergraph peeling processes. This is a widely-applied technique, e.g.\ in the design of efficient data structures like Invertible Bloom Filters~\cite{GoodrichM11}, cuckoo hashing~\cite{PaghR01,DietzfelbingerW07,Mitzenmacher09,FountoulakisP12,FriezeM12}, but also in the context of error-correcting codes~\cite{KarpLS04, LubyMSS01} and random SAT instances~\cite{Molloy04}. Before proceeding, we find it instructive to give a quick recap of this technique. The setup is that there are $K$ elements which are supposed to be hashed to $O(K)$ buckets in such a way that we can recover all $K$ elements given only the buckets. The main idea is to hash each element not just to a single bucket, but to $h > 1$ buckets, where each bucket stores the sum of the elements that are hashed to it. In order to recover the elements from the buckets, one repeatedly finds a bucket that contains exactly one item, recovers that one item and removes (i.e., subtracts) it from all the~$h$ bucket it participates in. This creates some new singletons, so the algorithm proceeds to recover these new singletons, and so on.

The analysis hinges on random (hyper-)graph theory. Specifically, one models the hashing scheme as an $h$-uniform hypergraph where each bucket is a vertex, and each element is a hyperedge spanning the~$h$ buckets it is hashed to. The goal is to show that this hypergraph is \emph{peelable}, i.e., that we can step-by-step remove edges containing some degree-1 node (i.e., a bucket containing just one element). Clearly, as the hash function is uniformly random, so is this hypergraph. And as the number of vertices (i.e., buckets) is larger than the number of hyperedge (i.e., elements) by a constant factor, the hypergraph is quite sparse. In such a situation one can prove that with constant probability all connected components are loosely connected and therefore peelable, at least as long as~\makebox{$h \geq 2$} (for $h \geq 3$ this event even happens with high probability).

This peeling process technique has already been used before in sparse recovery~\cite{Price11,EppsteinG07,LiYPPR19}. In~\cite{Price11} the support of the vector $x$ is known and thus the algorithm is concerned solely about estimation rather than the harder task of identification, whereas in~\cite{EppsteinG07,LiYPPR19} the authors assume that the vector is exactly $k$-sparse. In our case, where we assume nor knowledge of the support nor exact sparsity, analyzing this process is quite different and more challenging.

\paragraph{Weighted Hypergraph Peeling.}
Back to our context, let us attempt to mimic the Invertible Bloom Filter setup: We hash the~$n$ coordinates to $O(k/\epsilon)$ buckets in such a way that we can ideally recover all $O(k / \epsilon)$ heavy elements. Similarly to the Invertible Bloom Filter, we hash each element not just to a single bucket, but instead to some $h \geq 2$ buckets.\footnote{Our final algorithm succeeds only with constant probability, so it actually suffices to take $h = 2$ (i.e., to consider graphs). Following related literature we will nevertheless treat the general hypergraph case $h \geq 2$ here, also because we are optimistic that this generalization will pay off for future work (possibly to reduce the error probability in \cref{thm:main}).} We can augment each bucket with a scheme that can perform $1$-sparse recovery in $O(\log n)$ measurements and time with high probablity in $n$, i.e., whenever some $x_i$ is hashed to a bucket $v$ such that the energy $|x_i|^2$ significantly exceeds the total energy of all other entries hashed to the bucket~$v$ (by some large constant factor), then we can successfully recover $i$. On the side we keep an instance of Count-Sketch to have query access to a point-wise approximation~$\hat x$ of $x$~\cite{CharikarCF04}, and we keep a short linear sketch to approximate the tail~\makebox{$\hat t \approx \norm{x_{-k}}_2^2$}~\cite{NakosS19}. With these tools we simulate the Invertible Bloom Filter decoding process: For each bucket we run the $1$-sparse recovery procedure to find a candidate $i$. By comparing $\hat x_i$ with~$(\epsilon / k) \cdot \hat t$ we can reliably test if $i$ is heavy or light. If we determine that $i$ is light, the recovery failed and we discard $i$. If we determine that $i$ is heavy, then we store $i$ and subtract $\hat x_i$ from all buckets that $i$ participates in. The hope is that we can peel the heavy hitters, one by one, as in the Invertible Bloom Filter decoding procedure.

The analysis of this process naturally corresponds to a generalized \emph{weighted} hypergraph peeling process. Specifically, we model the recovery algorithm as a vertex- and edge-weighted $h$-uniform hypergraph $G = (V, E, w)$, where the nodes $V$ correspond to buckets and each edge $e$ corresponds to some heavy hitter $x_i$. We sometimes write $e_i$ to denote the edge corresponding to $x_i$. The edge weights are $w_{e_i} = |x_i|^2$ and the vertex weights are $w_v = \sum_i |x_i|^2$, where the sum is over all \emph{light} entries hashed to the bucket $v$. So, we have that $\mathbb{E}[w_v] = O(\frac{\epsilon}{k} \|x_{-k}\|_2^2)$. In this model, let us call an edge $e \in E$ \emph{free} if it is incident to some vertex $v \in e$ (i.e., bucket) such that
\begin{equation*}
    w_e \geq \rho \cdot \parens*{w_v + \sum_{\substack{e \neq e' \in E\\v \in e'}} w_{e'}},
\end{equation*}
where $\rho$ is a large constant called the \emph{recovery threshold}. We say that $e$ is \emph{peelable} if we can repeatedly remove free edges from the graph so that in the remaining graph $e$ itself becomes free.

This definition is indeed a generalization of the unweighted notion from before,\footnote{To see this, let all vertex weights be $0$, let all edge weights be $1$ and take a recovery threshold $\rho > 1$. Then an edge is free if and only if it is incident to a degree-$1$ node.} and it is easy to verify that it is chosen in such a way that our new algorithm can recover all peelable edges in~$G$. Therefore, ideally we would wish to prove that all edges in~$G$ are peelable and conclude the proof here. Unfortunately, there are two major problems in what we have outlined so far. In the following we describe these two problems and our solutions.

\paragraph{Problem 1: Not All Edges Are Peelable.}
The first serious problem is that, in contrast to the unweighted setting, it is unfortunately not true that all edges are peelable with high probability.\footnote{Namely, consider a heavy hitter with energy $|x_i|^2 \approx (\epsilon / k) \norm{x_{-k}}_2^2$. It corresponds to some edge $e_i$ with weight $w_{e_i} = |x_i|^2$. The total energy of the light elements is up to $\norm{x_{-k}}_2^2$, and thus the expected vertex weight $\Ex[w_v]$ in the hypergraph is roughly~\makebox{$\mu \approx (\epsilon / k) \norm{x_{-k}}_2^2$}. Therefore, with constant probability \emph{all} vertices $v$ incident to $e$ have weight~\makebox{$w_v \geq w_e$}, which effectively means that we can never recover $e$.} Does this mean that our approach is doomed? Recall that it is not strictly necessary to recover \emph{all} heavy hitters, but that we can afford to miss a small fraction (energy-wise). Therefore, the question is whether we can recover edges of sufficiently large total energy. 

Our analysis of this step turned out to be surprisingly clean. In fact, we only care about two aspects of $G$: (1) The hyperedges are random and sparse, and (2) the vertex weights are chosen randomly with expected weight $\mu$ (though they possibly depend on each other). In such a situation, we show that any edge $e$ fails to be peelable with probability at most $O(\mu / w_e)$.

The proof of this statement involves two steps (see \cref{lem:peeling-random-vertex-weights,lem:expected-spreadness}). First, fix the randomness of the edges (1) and focus on the randomness of the vertex weights (2). Intuitively, any large-weight edge $e$ only fails to be peelable if for all nodes $v$ incident to $e$, there are other non-peelable edges $e'$ of comparable total weight incident to $v$. But for these edges $e'$ the same argument applies. Hence, by induction $e$ only fails to be peelable if it is at the center of a sufficiently large and far-spreading connected component in which the weights slowly decrease as one moves away from $e$. Conversely, edges $e$ in smaller components should be peelable. Formally, this shows that each edge $e$ is peelable with probability $O(D(e) \cdot \mu / w_e)$ where $D(e) = \sum_{e'} \rho^{1+d(e, e')}$ measures how far-spread the connected component of $e$ behaves (the sum is over all edges $e'$ in the same connected component as $e$ and $d(e, e')$ is the distance between the edges $e$ and $e'$). In the second step we fix the randomness of the vertex weights (2) and only consider the randomness of the edges~(1). As long as the graph is sufficiently sparse, it can be shown that $\Ex[D(e)] = O(1)$. This concludes the proof that each edge is non-peelable with probability at most $O(\mu / w_e)$.

Getting back to our recovery algorithm, we demonstrate that this probability is indeed small enough so that the expected loss in energy is tolerable:
\begin{equation*}
    \mathbb{E}\left[\norm{x_{R \setminus T}}_2^2\right] = \sum_{i \in T} \Pr(\text{$e_i$ not peelable}) \cdot |x_i|^2 = \sum_{i \in T} O\parens*{\frac{\mu |x_i|^2}{w_{e_i}}} = O(|T| \mu) = O(\epsilon) \norm{x_{-k}}_2^2.
\end{equation*}

\paragraph{Problem 2: Accumulated Estimation Errors.}
There is a second serious problem with the approach outlined above. In the Invertible Bloom Filter decoder it is crucial that we can \emph{subtract} recovered elements from all relevant buckets---however, in our case this subtraction is \emph{not perfect}. Recall that instead of subtracting $x_i$ we can only subtract the Count-Sketch approximation~$\hat x_i$. However, this incurs an additive error of up to $\frac{\epsilon}{k} \cdot \norm{x_{-k}}_2^2$ to the energy of the respective bucket $v$ which affects all future recoveries involving this bucket $v$. Even worse, it could very well happen that these errors \emph{accumulate}. Concretely, in some buckets $v$ it could become impossible to perform recoveries---not because the noise $w_v$ is large, but because the accumulated approximation error of the previously recovered elements in that bucket $v$ became too large. In fact, we can tolerate an accumulated error of at most $O(\frac{\epsilon}{k} \cdot \norm{x_{-k}}_2^2)$, so even after a \emph{constant} number of recoveries per bucket, in the worst case, further recoveries become impossible.

Perhaps a first instinct would be to give up on these buckets, hoping again that the total energy loss is tolerable. Unfortunately, this does not work out at all---if we can only perform up to $O(1)$ recoveries per bucket then we effectively lose a constant fraction of the heavy hitters.

Instead, our solution relies on the tighter analysis of Count-Sketch due to Minton and Price~\cite{MintonP14}. They proved that a Count-Sketch with column sparsity $r = \Theta(\log n)$ and $O(r k / \epsilon)$ rows does not only give a point-wise approximation $\hat x$ satisfying that $|x_i - \hat x_i|^2 \leq \frac{\epsilon}{k} \cdot \norm{x_{-k}}_2^2$ (as in the original analysis of Count-Sketch~\cite{CharikarCF04}), but in fact provides an estimate that is more accurate by a factor of~$r$ \emph{on average.} Formally, the guarantee is that for all $i$,
\begin{equation*}
    \Pr\parens*{|x_i - \hat x_i|^2 \geq \frac{\lambda}{r} \cdot \frac{\epsilon}{k} \cdot \norm{x_{-k}}_2^2} \leq \exp(-\Omega(\lambda)).
\end{equation*}
see \cref{lem:count-sketch}. Intuitively, the estimation error $|x_i - \hat x_i|^2$ behaves like an exponentially distributed random variable. Combining this insight with the fact that with high probability the maximum bucket load from heavy hitters is \smash{$O(\frac{\log(k/\epsilon)}{\log\log(k/\epsilon)}) = O(\log n)$}, it follows that the expected accumulated error per bucket is $O(\frac{\epsilon}{k} \cdot \norm{x_{-k}}_2^2)$ (rather than the naive upper bound which is worse by the bucket load), which is tolerable in our setting. If the errors were independent of each other, a standard concentration bound for sums of independent subexponentially distributed random variables (e.g.,~\cite{Janson18}) would imply that the bound holds in fact with high probability. As the errors are dependent on each other our formal proof looks slightly different; see \cref{lem:point-estimates-intermediate,lem:bucket-error}. Curiously we depend on the tighter upper bound of \smash{$O(\frac{\log(k/\epsilon)}{\log\log(k/\epsilon)}) = O(\frac{\log n}{\log\log n})$} on the maximum bucket load.

\medskip

Our approach seems to inherently require $(k/\epsilon) \log n$ rows, instead of $(k/\epsilon) \log(n / k)$ as we need all the $1$-sparse recovery routines to succeed with high probability. With $\log (n/k)$ rows there is a $\poly(k/n)$ probability of not recovering a heavy hitter from a bucket, and this is not sufficiently small for the peeling process to go through. Nevertheless, the algorithm is very concise and is likely to be relevant in practice as well.

\paragraph{Outline.}
This concludes the high-level overview. In \cref{sec:toolkit} we describe and analyze our abstracted hypergraph peeling process, and describe the additional tools from sparse recovery and coding theory. In \cref{sec:sketch} we then give the complete sketch and formally analyze the peeling-based recovery algorithm.
\section{Toolkit} \label{sec:toolkit}
Throughout we write $[n] = \set{1, \dots, n}$. For a vector $x \in \Real^n$ and a set $S \subseteq [n]$, we let $x_S$ denote the vector obtained from $x$ by zeroing out all coordinates not in $S$. We also set $x_{-k}$ to denote the vector obtained from $x$ by zeroing the largest $k$ coordinates in magnitude (breaking ties arbitrarily but consistently).

\subsection{Hypergraph Peeling}
An \emph{$h$-uniform hypergraph} $G = (V, E)$ consists of \emph{vertices} $V$ and \emph{hyperedges} $E \subseteq \binom{V}{h}$ (i.e., each hyperedge is a size-$h$ set of vertices). A $2$-uniform hypergraph is simply a graph, and many of the basic concepts on graphs generalize to $h \geq 3$. For $u, v \in V$, a \emph{$u$-$v$-path} is a sequence of vertices~\makebox{$u = u_1, \dots, u_\ell = v$} such that each pair $\set{u_i, u_{i+1}}$ appears in a some hyperedge. We call~$u, v$ \emph{connected} if there exists a $u$-$v$-path, and we define their \emph{distance} $d(u, v)$ as the length of the shortest $u$-$v$-path. For sets $S, T \subseteq V$, we set $d(S, T) = \min_{s \in S, t \in T} d(s, t)$. The hypergraph $G$ is \emph{connected} if all pairs of nodes are connected, and we call a maximal connected subhypergraph a \emph{connected component}. We typically denote by $C(v)$ and~$C(e)$ the connected component containing a vertex~\makebox{$v \in V$} or a hyperedge $e \in E$.

A connected $h$-uniform hypergraph $G = (V, E)$ is a \emph{hypertree} if $|V| = (h - 1) |E| + 1$. It is easy to verify that this is the largest number of nodes in a connected hypergraph with $|E|$ hyperedges. If instead $|V| = (h - 1) |E|$ we call $G$ \emph{unicyclic}. See \cref{fig:hypertree-unicyclic} for an illustration.

\begin{figure} \label{fig:hypertree-unicyclic}
\caption{Illustrates a hypertree (left) and two unicyclic hypergraphs (middle and right).}
\vspace{-\smallskipamount}
\begin{adjustbox}{frame={\fboxrule} 18pt,minipage={\textwidth-2\fboxrule-36pt}}
\def\nodedist{1cm}
\def\edgesep{0.2cm}
\newcommand\drawedge[3]{
    \path[edg]
        let \p1 = (#1.center) in
        let \p2 = (#2.center) in
        let \p3 = (#3.center) in
        (\x2, \y2) + ({atan2(\y2 - \y1, \x2 - \x1)-90}:{\edgesep})
        arc ({atan2(\y2 - \y1, \x2 - \x1)-90}:{atan2(\y3 - \y2, \x3 - \x2)-90}:{\edgesep})
        -- ++(\x3 - \x2, \y3 - \y2)
        arc ({atan2(\y3 - \y2, \x3 - \x2)-90}:{atan2(\y1 - \y3, \x1 - \x3)-90}:{\edgesep})
        -- ++(\x1 - \x3, \y1 - \y3)
        arc ({atan2(\y1 - \y3, \x1 - \x3)-90}:{atan2(\y2 - \y1, \x2 - \x1)+270}:{\edgesep})
        -- cycle;}
\tikzset{
    vtx/.style={fill, circle, inner sep=1pt},
    edg/.style={draw},
}
\begin{tikzpicture}
\node[vtx] (a1) {};
\node[vtx, below left=0.9cm and 0.5cm of a1] (a2) {};
\node[vtx, below right=0.9cm and 0.5cm of a1] (a3) {};
\node[vtx, below left=0.9cm and -0.4cm of a3] (a4) {};
\node[vtx, below right=0.9cm and 1.2cm of a3] (a5) {};
\node[vtx, below left=0.9cm and 0.2cm of a2] (a6) {};
\node[vtx, below right=0.9cm and 0.6cm of a2] (a7) {};
\node[vtx, below left=0.9cm and 2.0cm of a2] (a8) {};
\node[vtx, below right=0.9cm and -1.2cm of a2] (a9) {};
\node[vtx, below left=0.9cm and 0.4cm of a6] (a10) {};
\node[vtx, below right=0.9cm and 0.4cm of a6] (a11) {};
\drawedge{a1}{a2}{a3}
\drawedge{a3}{a4}{a5}
\drawedge{a2}{a6}{a7}
\drawedge{a2}{a8}{a9}
\drawedge{a6}{a10}{a11}
\end{tikzpicture}
\hfill
\begin{tikzpicture}
\node[vtx] (a1) {};
\node[vtx, below left=0.9cm and 0.5cm of a1] (a2) {};
\node[vtx, below right=0.9cm and 0.5cm of a1] (a3) {};
\node[vtx, below left=0.9cm and -0.4cm of a3] (a4) {};
\node[vtx, below right=0.9cm and 1.2cm of a3] (a5) {};
\node[vtx, below left=0.9cm and 0.2cm of a2] (a6) {};
\node[vtx, below right=0.9cm and 0.6cm of a2] (a7) {};
\node[vtx, below left=0.9cm and 2.0cm of a2] (a8) {};
\node[vtx, below right=0.9cm and -1.2cm of a2] (a9) {};
\node[vtx, below left=0.9cm and 0.4cm of a6] (a10) {};
\node[vtx, below right=0.9cm and 0.4cm of a6] (a11) {};
\node[vtx, below right=0.9cm and 0.4cm of a7] (a12) {};
\drawedge{a1}{a2}{a3}
\drawedge{a3}{a4}{a5}
\drawedge{a2}{a6}{a7}
\drawedge{a2}{a8}{a9}
\drawedge{a6}{a10}{a11}
\drawedge{a7}{a11}{a12}
\end{tikzpicture}
\hfill
\begin{tikzpicture}
\def\rin{1.1cm}
\def\rout{1.7cm}
\node[vtx] (b1) at (90:{\rin}) {};
\node[vtx] (b2) at (162:{\rin}) {};
\node[vtx] (b3) at (234:{\rin}) {};
\node[vtx] (b4) at (306:{\rin}) {};
\node[vtx] (b5) at (18:{\rin}) {};
\node[vtx] (b6) at (126:{\rout}) {};
\node[vtx] (b7) at (198:{\rout}) {};
\node[vtx] (b8) at (270:{\rout}) {};
\node[vtx] (b9) at (342:{\rout}) {};
\node[vtx] (b10) at (54:{\rout}) {};
\drawedge{b6}{b2}{b1}
\drawedge{b2}{b7}{b3}
\drawedge{b3}{b8}{b4}
\drawedge{b5}{b4}{b9}
\drawedge{b10}{b1}{b5}
\end{tikzpicture}
\end{adjustbox}
\end{figure}
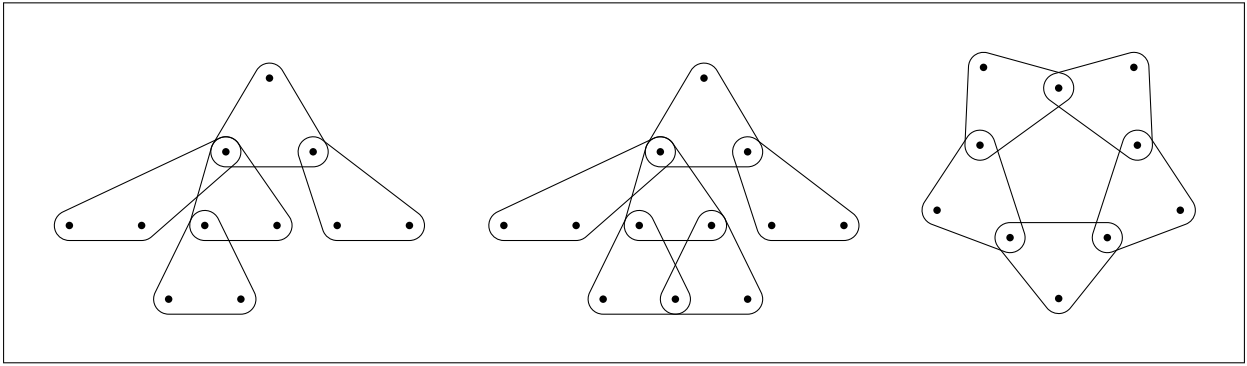

\paragraph{Unweighted Hypergraph Peeling.}
We start with two lemmas that are useful in the analysis of unweighted peeling processes. Throughout we consider \emph{uniformly random} hypergraphs that are sampled with repetition---i.e., starting from an empty graph on $N$ nodes, we insert $M$ uniformly random $h$-uniform hyperedges into the graph. It is possible for the same edge to be sampled multiple times, so the resulting object is strictly speaking a \emph{multi-}hypergraph. But this event happens only with small probability in our parameter regime, and thus we will often implicitly assume simple hypergraphs. For a proof of the following lemma, see e.g.\ the text book~\cite[Lemma 12.3]{FriezeK15}.

\begin{lemma} \label{lem:hypertrees-unicyclic}
Let $h \geq 3$ and let $G$ be a $h$-uniform hypergraph with $N$ vertices and~\makebox{$M \leq \frac{N}{2h}$} hyperedges sampled uniformly at random (with repetition). Then with probability \smash{$1 - N^{-1+o(1)}$} each connected component is a hypertree or unicyclic.
\end{lemma}

Notably the same statement also holds for graphs (i.e., $h = 2$), but only with \emph{constant} success probability~\cite[Theorem~2.1]{FriezeK15}. From this lemma it is almost immediate that random hypergraphs are peelable with high probability, i.e., that we can repeatedly remove edges involved incident to a degree-1 vertex to remove all edges from the hypergraph.  Indeed, this is clear for hypertrees (simply keep removing leaf edges) and is also not hard to check for unicyclic components. In the same spirit, the following lemma will be useful for the treatment of unicyclic components later.

\begin{lemma} \label{lem:unicyclic-removal}
Let $G = (V, E)$ be a unicyclic hypergraph. For each edge $e \in E$ there is some vertex $v \in e$ such that upon removing $e$ the resulting connected component $C_{G \setminus e}(v)$ is a hypertree.
\end{lemma}
\begin{proof}
Let $N$ and $M$ denote the number of vertices and edges in the connected component $C(e)$, respectively. Since the graph is unicyclic we have that $N = (h-1) M$. Now consider the graph~$G \setminus e$ obtained from $G$ by removing $e$, and let $C_1, \dots, C_\ell$ denote its connected components. Suppose for contradiction that all these components are not hypertrees. Letting $N_i$ and $M_i$ denote the number of vertices and edges in $C_i$, respectively, this means that $N_i \leq (h-1) M_i$. But clearly $\sum_i N_i = N$ and $\sum_i M_i = M - 1$, which contradicts our initial assumption that $N = (h-1) M$.
\end{proof}

\paragraph{Weighted Hypergraph Peeling.}
Finally, we consider our new weighted hypergraph peeling processes. Formally, equip a hypergraph $G = (V, E)$ with a weight function $w : V \sqcup E \to \Real_{\geq 0}$. We will write $w_v$ (for $v \in V$) and $w_e$ (for $e \in E$) to refer to the vertex- and edge-weights, respectively. We first define weighted hypergraph peeling formally:

\begin{definition}[Peeling]
Let $\rho \geq 1$ and let $G = (V, E, w)$ be a vertex- and edge-weighted hypergraph. We say that an edge $e \in E$ is \emph{$\rho$-free} if there is some vertex $v \in e$ such that
\begin{equation*}
    w_e \geq \rho \cdot \left(w_v + \sum_{e \neq e' \in E} w_{e'}\right).
\end{equation*}
A \emph{$\rho$-peeling sequence} is a sequence of edges $e_1, \dots, e_\ell$ such that, for all $i \in [\ell]$, the edge $e_i$ is $\rho$\=/free in the graph $G \setminus \set{e_1, \dots, e_{\ell-1}}$. We call $e$ \emph{$\rho$-peelable} if there is a $\rho$-peeling sequence containing~$e$.
\end{definition}

The next two lemmas give convenient tools to analyze weighted hypergraph peeling. In combination, they essentially state that the probability that any edge is not peelable is at most $O(\mu / w_e)$, where $\mu$ is the expected vertex weight.

\begin{lemma}[Peeling with Random Vertex Weights] \label{lem:peeling-random-vertex-weights}
Fix an arbitrary hypergraph $G = (V, E)$ and arbitrary edge weights $w_e$, and suppose that the vertex weights $w_v$ are chosen randomly such that $\Ex[w_v] \leq \mu$ (though these weights might be dependent). Let $e \in E$ be such that $C(e)$ is a hypertree or unicyclic. The probability that $e$ is not $\rho$-peelable is at most \smash{$\frac{\mu \cdot D(e)}{w_e}$} where
\begin{equation*}
    D(e) = D_{G, \rho}(e) = \sum_{e' \in E(C(e))} \rho^{1+d(e, e')}.
\end{equation*}
\end{lemma}
\begin{proof}
Let $e^* \in E$ be an arbitrary edge. First assume that the connected component $C(e^*)$ is a hypertree. Let $v^* \in e^*$ be arbitrary. We will consider $C(e^*)$ as a \emph{rooted} hypertree $T$ with root~$v^*$ in the following natural way: Each edge $e$ picks a root node $r(e) \in e$ which is the unique first node discovered by performing a breadth-first search from $v^*$. Let $\delta(v) = \set{e : r(e) = v}$ denote the set of edges rooted directly at $v$. Let $\Delta(v)$ denote the set of edges in the entire subtree rooted at $v$, and let $\Delta(e) = \Delta(r(e))$.

We construct a peeling sequence by traversing the hypertree $T$ from bottom to top and attempting to peel all edges. Formally, consider the following recursive process $\textsc{Peel}(v)$ to construct a peeling sequence $\sigma$:

\bigskip
\begin{adjustbox}{frame={\fboxrule} 8pt,minipage={\textwidth-2\fboxrule-16pt}}
\begin{algorithmic}[1]
\Procedure{Peel}{$v$}
    \State Let $\sigma$ be an empty peeling sequence
    \ForEach{$e \in \delta(v)$}
        \ForEach{$v \neq u \in e$}
            \State Append $\textsc{Peel}(u)$ to $\sigma$
        \EndForEach
        \If{$e$ is free in $T \setminus \sigma$}
            \State Append $e$ to $\sigma$
        \EndIf
    \EndForEach
    \State\Return $\sigma$
\EndProcedure
\end{algorithmic}
\end{adjustbox}
\bigskip

Consider the peeling sequence constructed by $\textsc{Peel}(v^*)$. Clearly, as both the hypergraph $G$ and the edge weights are fixed, this sequence only depends on the randomness of the vertex-weights. Let $X_e \in \set{0, 1}$ be the random variable indicating if the edge $e$ does not appear in this peeling sequence. By induction we show that
\begin{equation*}
    \Pr(X_e = 1) \leq \frac{\mu}{w_e} \cdot \sum_{e' \in \Delta(e)} \rho^{1+d(e, e')}.
\end{equation*}

To prove this, take an arbitrary edge $e$ and pick any non-root node $v \in e$ (i.e. $v \neq r(e)$). By construction we can include $e$ in the peeling sequence if $e$ is free in the graph where we have removed all edges in the recursively constructed peeling sequences $\textsc{Peel}(u)$, where $u$ ranges over the children of $v$. In other words, if
\begin{equation*}
    w_e \geq \rho \cdot \parens*{w_v + \sum_{e' \in \delta(v)} X_{e'} w_{e'}}.
\end{equation*}
(Note that for leaf nodes $v$ the set $\delta(v)$ is empty and this condition simplifies to $w_e \geq \rho \cdot w_v$.) By Markov's inequality we thus have
\begin{align*}
    \Pr(X_e = 1)
    &\leq \Pr\parens*{\rho \cdot \parens*{w_v + \sum_{e' \in \delta(v)} X_{e'} w_{e'}} > w_e} \\
    &\leq \frac{\rho}{w_e} \cdot \Ex\brackets*{w_v + \sum_{e' \in \delta(v)} X_{e'} w_{e'}} \\
    &= \frac{\rho}{w_e} \cdot \parens*{\Ex[w_v] + \sum_{e' \in \delta(v)} \Ex[X_{e'}] \cdot w_{e'}}
\intertext{We emphasize that here we only use linearity of expectation and therefore any dependencies between the variables $w_v$ and the resulting random variables $X_{e'}$ are irrelevant. Using that $\Ex[w_v] \leq \mu$ and using the induction hypothesis we can further bound this expression as follows:}
    &\leq \frac{\rho}{w_e} \cdot \parens*{\mu + \sum_{e' \in \delta(v)} \frac{w_{e'} \cdot \mu}{w_{e'}} \sum_{e'' \in \Delta(e')} \rho^{1 + d(e', e'')}} \\
    &\leq \frac{\mu}{w_e} \cdot \parens*{\rho + \rho \cdot \sum_{e' \in \Delta(e) \setminus \set{e}} \rho^{d(e, e')}} \\
    &\leq \frac{\mu}{w_e} \cdot \parens*{\rho + \sum_{e' \in \Delta(e) \setminus \set{e}} \rho^{1+d(e, e')}} \\
    &\leq \frac{\mu}{w_e} \cdot \sum_{e' \in \Delta(e)} \rho^{1+d(e, e')}.
\end{align*}
All in all, we have shown that $\Pr(X_{e^*} = 1)$ is bounded as in the lemma statement whenever $C(e^*)$ is a hypertree.

Finally, suppose that $C(e^*)$ is unicyclic. By \cref{lem:unicyclic-removal} there is some vertex $v^* \in e^*$ such that when removing $e^*$ the resulting connected component $C_{G \setminus e}(v^*)$ is a hypertree $T$. By the same inductive argument we can construct a peeling sequence $\sigma$ that fails to peel any edge $e$ in~$T$ with probability at most~\smash{$\Pr(X_e = 1) \leq \frac{\mu}{w_e} \cdot \sum_{e' \in \Delta(e)} \rho^{1+d(e, e')}$}. If $e^*$ is $\rho$-free in the graph $G \setminus \sigma$ then clearly $e^*$ is $\rho$-peelable. By exactly the same argument as before, $e^*$ fails to be $\rho$-free with probability at most
\begin{equation*}
    \Pr\parens*{\rho \cdot \parens*{w_{v^*} + \sum_{e' \in \delta(v)} X_{e'} w_{e'}} > w_{e^*}} \leq \frac{\mu}{w_{e^*}} \cdot \sum_{e \in T} \rho^{1 + d(e^*, e)},
\end{equation*}
which completes the proof.
\end{proof}

\begin{lemma} \label{lem:expected-spreadness}
Let $G = (V, E)$ be an $h$-uniform hypergraph where we fix one hyperedge $e \in E$ and sample~\smash{$M \leq \frac{|V|}{8 \rho h^2}$} other hyperedges uniformly at random (with repetition). Then $\Ex[D_{G, \rho}(e)] \leq \rho + 1$.
\end{lemma}
\begin{proof}
To bound $D_{G, \rho}(e)$ we perform a breadth-first search starting from $e$. Let $E_0, E_1, \dots$ denote the sets of edges explored in the respective layers of the search. Formally, let $E_0 = \set{e}$ and for~\makebox{$i \geq 1$} the set $E_i$ consists of all edges $e' \not\in E_1, \dots, E_{i-1}$ for which there is an edge $e'' \in E_{i-1}$ that shares a common node with $e'$. We view the random construction of $G$ as being revealed step-by-step in the breadth-first search. I.e., initially we only know that $e$ is an edge but have not revealed any random choice, and in the $i$-th step we reveal only the edges touching $E_i$.

We will bound~$\Ex[|E_i| \mid E_{i-1}]$. So focus on the $i$-th step where we have already revealed $E_{i-1}$ and are currently revealing the edges $E_i$. The only candidates are $\set{v_1, \dots, v_h}$ where at least one of these nodes touches an edge in $E_{i-1}$. There are at most $\binom{|V|}{h-1} \cdot h |E_{i-1}|$ such sets. Moreover, at least $|V| - M \geq (1 - \frac{1}{2h}) |V|$ nodes are not touched by any of the previously revealed edges. Thus, any fix $\set{v_1, \dots, v_h}$ becomes an edge in this step with probability at most
\begin{equation*}
    \frac{M}{\binom{|V|-M}{h}} \leq \frac{M}{\binom{|V|}{h}} \cdot \parens*{\frac{|V|}{|V|-M}}^h = \frac{M}{\binom{|V|}{h}} \cdot \parens*{\frac{1}{1-\frac{1}{2h}}}^h \leq \frac{2M}{\binom{|V|}{h}}.
\end{equation*}
Thus,
\begin{align*}
    \Ex[|E_i| \mid E_{i-1}]
    &\leq \binom{|V|}{h-1} \cdot h |E_{i-1}| \cdot \frac{2M}{\binom{|V|}{h}} \\
    &\leq \binom{|V|}{h-1} \cdot h |E_{i-1}| \cdot \frac{|V|}{4 \rho h^2 \binom{|V|}{h}} \\
    &= \frac{|E_{i-1}|}{4\rho} \cdot \frac{|V|}{|V| - (h - 1)} \\
    &\leq \frac{|E_{i-1}|}{2\rho},
\end{align*}
and it follows that
\begin{equation*}
    \Ex[|E_i|] \leq \frac{1}{(2\rho)^i}.
\end{equation*}
Clearly all edges $e' \in E_i$ have distance at most $i-1$ to the designated edge $e$, and therefore
\begin{equation*}
    \Ex[D_{G', \rho}(e)] \leq \Ex\brackets*{\rho + \sum_{i=1}^\infty |E_i| \cdot \rho^i} \leq \rho + \sum_{i=1}^\infty \frac{\rho^i}{(2\rho)^i} = \rho + 1,
\end{equation*}
completing the proof.
\end{proof}

\subsection{Sparse Recovery}

\begin{lemma}[Count-Sketch~\cite{CharikarCF04,MintonP14}] \label{lem:count-sketch}
For all integers $r \geq 1$ and $2 \leq s \leq n$, there is a randomized linear sketch $\Phi \in \Real^{m \times n}$ with $m = O(s r)$ rows and column sparsity $O(r)$ such that, given $y = \Phi x$ and some coordinate $i \in [n]$, we can compute in time $O(r)$ an approximation $\hat x_i$ satisfying that
\begin{equation*}
    \Pr\left(|x_i - \hat x_i|^2 > \frac{\lambda}{r} \cdot \frac{\norm{x_{-s}}_2^2}{s}\right) \leq 2 \exp(-\lambda),
\end{equation*}
for any $\lambda \leq r$.
\end{lemma}

\begin{lemma}[Tail Estimation Sketch~\cite{NakosS19}] \label{lem:tail-sketch}
There is a constant $c \geq 1$, so that for all integers~\makebox{$k \leq n$} there is a randomized linear sketch $\Phi \in \mathbb{R}^{m\times n}$ with $m = O(\log n)$ rows and column sparsity $O(\log n)$ such that given $y = \Phi x$ we can find in time $O(\log n)$ a number $\tilde t$ satisfying with high probability that
\begin{equation*}
    \frac{1}{c k} \cdot \norm{x_{-c k}}_2^2 \leq \tilde t \leq \frac{1}{k} \cdot \norm{x_{-k}}_2^2.
\end{equation*}
\end{lemma}

\subsection{Error-Correcting Codes}
As the third ingredient we rely on an asymptotically good error-correcting code that can be encoded and decoded in \emph{linear} time. Such a result was first obtained by Spielman~\cite{Spielman96}, and improved constructions have since lead to improved constants. Specifically, we rely on the following lemma which follows from the error-correcting codes due to Guruswami and Indyk~\cite{GuruswamiI05}.\footnote{In the language of coding theory, the function $\operatorname{Enc}$ is a binary error-correcting code with rate $R := \frac{1}{2048}$ that can be decoded (via $\operatorname{Dec}$) in linear time from a fraction of up to $e := \frac{1}{8}$ errors. In~\cite[Theorem~5]{GuruswamiI05} it is shown that any rate-distance trade-off below the Zyablov bound is achievable; that is, we have that $e \leq \max_{R < r < 1} (1 - r - \epsilon) H^{-1}(1 - R / r)$ (for some constant $\epsilon > 0$), where~$H^{-1}$ is the inverse of the binary entropy function.

The second item in \cref{lem:code} is non-standard, but can easily be ensured by flipping each bit in $\operatorname{Enc}(\cdot)$ with probability $\frac{1}{2}$. By Chernoff's bound, each resulting code word differs from the expected number of 1's, $1024 \log n$, by a fraction of more than $\epsilon := \frac{1}{16}$ with probability at most~\smash{$\exp(-\frac{\epsilon^2}{2} \cdot 1024 \log n) \leq n^{-2}$}, and so the statement follows by a union bound over all $n$ code words.}

\begin{lemma}[Linear-Time Error-Correcting Code~\cite{GuruswamiI05}] \label{lem:code}
There is a randomized algorithm to evaluate functions $\operatorname{Enc} : [n] \to \set{0, 1}^{2048 \log n}$ and $\operatorname{Dec} : \set{0, 1}^{2048 \log n} \to [n]$ with the following properties:
\begin{itemize}
    \item Let $i \in [n]$ and let $x \in \set{0, 1}^{2048 \log n}$ be such that $\operatorname{Enc}(i)$ differs from $x$ in at most a $\frac{1}{8}$-fraction of the symbols. Then $\operatorname{Dec}(x) = i$.
    \item With high probability, the fraction of 1's in $\operatorname{Enc}(i)$ is at least $\frac{1}{2} - \frac{1}{16}$ and at most $\frac{1}{2} + \frac{1}{16}$ for all $i \in [n]$.
    \item Evaluating $\operatorname{Enc}(\cdot)$ and $\operatorname{Dec}(\cdot)$ takes time $O(\log n)$.
\end{itemize}
\end{lemma}

\section{The Sketch in Detail} \label{sec:sketch}
In this section we provide a formal proof of our main theorem:

\thmmain*

For concreteness, we will demonstrate the theorem with error probability $0.1$ throughout, but any smaller constant success probability can easily be achieved by choosing the involved constants appropriately larger.
 
We begin with a detailed description of the sketch. It consists of three different types of queries summarized in the following (that is, the rows of the sketching matrix can be partitioned into three blocks corresponding to the following three types of queries). Here, as in the theorem statement we denote by $x$ denote the queried vector.
\begin{itemize}
    \item By \cref{lem:tail-sketch} we estimate the tail $\frac{1}{c} \norm{x_{-c k}}_2^2 \leq \hat t \leq \norm{x_{-k}}_2^2$.
    \item By \cref{lem:count-sketch} (with parameters $r = 10 \log n$ and $s = \ceil{8192 c k / \epsilon}$) we obtain point-wise query access to an estimate $\hat x$ of $x$.
    \item The main queries are described by the following construction. Let $V = [2^{19} c k / \epsilon]$; we will informally refer to $V$ as the set of \emph{buckets} and refer to any $v \in V$ as a \emph{bucket.} For each coordinate $i \in [n]$, we sample a uniformly random size-3 set~\makebox{$e_i \subseteq V$}; in other words we hash each coordinate to exactly 3 buckets. We write $B_v = \set{i \in [n] : v \in e_i}$ to denote the set of buckets that $i$ was hashed to. Let $\sigma_1, \dots, \sigma_{2048 \log n} : [n] \to \set{-1, 1}$ be uniformly random. We query, for each bucket $v \in V$ and for each $j \in [2048 \log n]$, the values
    \begin{equation*}
        q_{v, j} = \sum_{i \in B_v} \sigma_j(i) \cdot \operatorname{Enc}(i)_j \cdot x_i.
    \end{equation*}
\end{itemize}
It is easy to verify that these queries can be implemented by a sketching matrix $A \in \Real^{m \times n}$ with $m = O((k / \epsilon) \log n)$ rows and with column sparsity $O(\log n)$ (see \cref{lem:column-sparsity}).

\bigskip
\begin{adjustbox}{frame={\fboxrule} 8pt,minipage={\textwidth-2\fboxrule-16pt}}
\begin{algorithmic}[1]
    \State Initialize $Q \gets V$ and $R \gets \emptyset$
    \While{$Q \neq \emptyset$} \label{alg:recover:line:loop}
        \State Take an arbitrary bucket $v \in Q$ and remove $v$ from $Q$ \label{alg:recover:line:pop}
        \State Compute the following values $p_{v, j}$ for all $j \in [2048 \log n]$: \label{alg:recover:line:residual}
        \begin{equation*}
            \quad p_{v, j} = q_{v, j} - \sum_{i \in B_v \cap R} \sigma_j(i) \cdot \operatorname{Enc}(i)_j \cdot \hat x_i
        \end{equation*}
        \State Let $y \in \set{0, 1}^{2048 \log n}$ be the string with 1's in the largest half of the entries of $(p_{v, j})_j$ \label{alg:recover:line:median}
        \State Compute $i \gets \operatorname{Dec}(y)$ \label{alg:recover:line:dec}
        \State Query Count-Sketch to obtain $\hat x_i$
        \If{$|\hat x_i|^2 \geq \frac{\epsilon}{2k} \cdot \hat t$ \AND{} $i \not\in R$} \label{alg:recover:line:threshold}
            \State Update $Q \gets Q \cup e_i$ and $R \gets R  \cup \set{i}$ \label{alg:recover:line:push}
        \EndIf
    \EndWhile
    \State Sort $R = \set{i_1, \dots, i_{|R|}}$ such that $|\hat x_{i_1}| \geq \dots \geq |\hat x_{i_{|R|}}|$
    \State Let $S = \set{i_1, \dots, i_{\min(3k, |R|)}}$
    \State\Return $x' \gets \hat x_S$
\end{algorithmic}
\end{adjustbox}
\bigskip

We analyze the sketch in seven steps: We fix some handy notation (\cref{sec:sketch:sec:notation}), state the success events for the point and tail estimates (\cref{sec:sketch:sec:point-tail-estimates}), prove that the estimation error in each bucket is small (\cref{sec:sketch:sec:bucket-estimates}), model the recovery algorithm as a weighted hypergraph peeling process (\cref{sec:sketch:sec:peeling}), show that the significant mass of the top-$k$ coordinates is peelable and thus recovered (\cref{sec:sketch:sec:peelable}), analyze the final refinement step and conclude the overall correctness of the sketch (\cref{sec:sketch:sec:correctness}), and finally analyze the time complexity of the recovery algorithm (\cref{sec:sketch:sec:time}).

\subsection{Notation} \label{sec:sketch:sec:notation}
Throughout we fix the vector $x$.

\begin{definition}[Top Coordinates] \label{def:top-k}
Let $T \subseteq [n]$ denote the set of the largest $k$ coordinates in $x$ (in magnitude), breaking ties arbitrarily.
\end{definition}

\begin{definition}[Heavy/Light/Intermediate Coordinates] \label{def:heavy-light-intermediate}
We call $i \in [n]$ \emph{heavy} if~\smash{$|x_i|^2 \geq \frac{\epsilon}{k} \norm{x_{-k}}_2^2$} and \emph{light} otherwise. We call $i \in [n]$ \emph{intermediate} if~\smash{$|x_i|^2 \geq \frac{\epsilon}{4c k} \cdot \norm{x_{-c k}}_2^2$}. Let $H, L, I \subseteq [n]$ denote the sets of heavy, light, and intermediate coordinates, respectively.
\end{definition}

That is, the coordinates are perfectly partitioned into heavy and light, whereas the intermediate coordinates form a superset of the heavy coordinates. We also write $H_v = B_v \cap H$, $L_v = B_v \cap L$ and $I_v = B_v \cap I$ to express the heavy/light/intermediate coordinates hashed to the bucket $v$.

Throughout we assume the crude upper bound $k / \epsilon = O(n)$, as otherwise we can solve the problem trivially using $n$ measurements.

\begin{table}[t]
\caption{List of symbols.}
\vspace{-\smallskipamount}
\begin{adjustbox}{frame={\fboxrule} 8pt,minipage={\textwidth-2\fboxrule-16pt}}
\begin{tabularx}{\linewidth}{lXl}
    \emph{Symbol} & \emph{Description} & \emph{Reference} \\[\smallskipamount]
    $n$ & Dimension & \cref{thm:main} \\
    $k$ & Sparsity threshold & \cref{thm:main} \\
    $\epsilon$ & Error tolerance & \cref{thm:main} \\
    $x$ & Query vector & \cref{thm:main} \\
    $x'$ & Output vector & \cref{thm:main} \\
    $\hat x$ & Estimates $x$ point-wise by Count-Sketch & \cref{lem:count-sketch} \\
    $\hat t$ & Estimates the tail $\norm{x_{-k}}_2^2$ & \cref{lem:tail-sketch} \\
    $c$ & Constant from \cref{lem:tail-sketch} & \cref{lem:tail-sketch} \\
    $T$ & Set of largest $k$ coordinates in $x$ & \cref{def:top-k} \\
    $H$ & Set of heavy coordinates in $x$ & \cref{def:heavy-light-intermediate} \\
    $L$ & Set of light coordinates in $x$ & \cref{def:heavy-light-intermediate} \\
    $I$ & Set of intermediate coordinates in $x$ & \cref{def:heavy-light-intermediate} \\
    $I_\lambda$ & Set of intermediate coordinates misestimated by a factor of $\lambda$ & \cref{lem:point-estimates-intermediate} \\
    $B_v$ & Set of coordinates hashed to bucket $v$ \\
    $H_v$ & Set of heavy coordinates hashed to bucket $v$ \\
    $L_v$ & Set of light coordinates hashed to bucket $v$ \\
    $I_v$ & Set of intermediate coordinates hashed to bucket $v$ \\
    $G$ & Associated hypergraph & \cref{def:associated-hypergraph} \\
    $V$ & Set of buckets a.k.a.\ vertices in the associated hypergraph & \cref{def:associated-hypergraph} \\
    $E$ & Set of edges in the associated hypergraph & \cref{def:associated-hypergraph} \\
    $e_i$ & Defines to which 3 buckets a coordinate $i$ is hashed \\
    $w$ & Vertex- and edge-weights of the associated hypergraph & \cref{def:associated-hypergraph} \\
    $R$ & Set of recovered coordinates \\
    $S$ & Set of $\min(3k, |R|)$ recovered coordinates $i$ with largest $|\hat x_i|$ \\
    $\rho$ & Peeling threshold ($\rho = 2048$) \\
    $\operatorname{Enc}$ & Encoding function of the error-correcting code & \cref{lem:code} \\
    $\operatorname{Dec}$ & Decoding function of the error-correcting code & \cref{lem:code} \\
\end{tabularx}
\end{adjustbox}
\end{table}

\begin{lemma}[Number of Heavy Coordinates] \label{lem:heavy-intermediate-number}
$|H| \leq 2k / \epsilon$ and $|I| \leq 8 c k / \epsilon$.
\end{lemma}
\begin{proof}
Clearly~\smash{$\norm{x_{H \setminus T}}_2^2 \leq \norm{x_{-k}}_2^2$}, and thus there can be at most $k / \epsilon$ elements $i \in H \setminus T$ with magnitude~\smash{$|x_i|^2 \geq \frac{\epsilon}{k} \norm{x_{-k}}_2^2$}. It follows that $|H| \leq k/\epsilon + |T| = k/\epsilon + k \leq 2k/ \epsilon$. The bound on $|I|$ is analogous.
\end{proof}

Our ultimate goal is to recover a subset of the top-$k$ coordinates which suffices for $\ell_2/\ell_2$ sparse recovery. In our one-shot recovery, we settle to recover a subset of the heavy coordinates which contains a desired subset as above. However, due to the approximate nature of the algorithm, it is likely that the algorithm will recover some coordinates that are not heavy but a constant fraction smaller than the threshold of heavy. We call these coordinates intermediate.

\subsection{Point and Tail Estimates} \label{sec:sketch:sec:point-tail-estimates}
The first step is to analyze the accuracy of the point and tail estimates from the first two types of queries in the sketch. The following two lemmas are immediate from \cref{lem:count-sketch,lem:tail-sketch}, but we record both for convenience later on.

\begin{lemma}[Tail Estimates] \label{lem:tail-estimates}
With high probability,~\smash{$\frac{1}{c} \norm{x_{-c k}}_2^2 \leq \hat t \leq \norm{x_{-k}}_2^2$}.
\end{lemma}

\begin{lemma}[Point Estimates] \label{lem:point-estimates}
With high probability, $|x_i - \hat x_i|^2 \leq \frac{\epsilon}{32 c k} \cdot \norm{x_{-c k}}_2^2$ for all $i \in [n]$.
\end{lemma}

The first consequence is that we can pinpoint which coordinates $i$ can pass the test ``$|x_i|^2 \geq \frac{\epsilon}{2k} \cdot \hat t$\,'' in \cref*{alg:recover:line:threshold} of the recovery algorithm: In the following lemma we show that \emph{all} heavy coordinates pass the test, and the test is \emph{only} passed by intermediate coordinates.

\begin{lemma}[Recovery Threshold] \label{lem:recovery-threshold}
Suppose that the events from \cref{lem:tail-estimates,lem:point-estimates} hold. Then for all $i \in H$ it holds that $|\hat x_i|^2 \geq \frac{\epsilon}{2k} \hat t$, whereas for all $i \not\in I$ it holds that $|\hat x_i|^2 < \frac{\epsilon}{2k} \hat t$.
\end{lemma}
\begin{proof}
Conditioned on these events we have~\smash{$\frac{1}{c} \norm{x_{-c k}}_2^2 \leq \hat t \leq \norm{x_{-k}}_2^2$} and~\smash{$|x_i - \hat x_i|^2 \leq \frac{\epsilon}{32 c k} \cdot \norm{x_{-c k}}_2^2$} for all $i \in [n]$. In this case, for each $i \in H$ we have
\begin{equation*}
    |\hat x_i| \geq |x_i| - |x_i - \hat x_i| \geq \sqrt{\frac{\epsilon \norm{x_{-k}}_2^2}{k}} - \sqrt{\frac{\epsilon \norm{x_{-ck}}_2^2}{32 c k}} \geq \sqrt{\frac{\epsilon \hat t}{k}} - \sqrt{\frac{\epsilon \hat t}{32 k}} \geq \sqrt{\frac{\epsilon \hat t}{2k}},
\end{equation*}
and for each $i \not\in I$ we have
\begin{equation*}
    |\hat x_i| \leq |x_i| + |x_i - \hat x_i| < \sqrt{\frac{\epsilon \norm{x_{-ck}}_2^2}{4ck}} + \sqrt{\frac{\epsilon \norm{x_{-ck}}_2^2}{32 c k}} \leq \sqrt{\frac{\epsilon \hat t}{4k}} + \sqrt{\frac{\epsilon \hat t}{32 k}} \leq \sqrt{\frac{\epsilon \hat t}{2 k}},
\end{equation*}
as claimed.
\end{proof}

The following lemma is more interesting. We exploit that Count-Sketch in fact provides point estimates that are better by a factor of $\log n$ \emph{on average} as proven by Minton and Price~\cite{MintonP14} to argue that with good probability the fraction of (intermediate) coordinates that are misestimated by a factor of $\Theta(\lambda)$ is exponentially small in $\lambda$.

\begin{lemma}[Intermediate Point Estimates] \label{lem:point-estimates-intermediate}
For an integer $\lambda \geq 1$, let $I_\lambda \subseteq I$ denote the subset of coordinates $i$ with
\begin{equation*}
    \frac{\epsilon(\lambda-1)}{2^{14} k \log n} \cdot \norm{x_{-k}}_2^2 \leq |x_i - \hat x_i|^2 \leq \frac{\epsilon \lambda}{2^{14} k \log n} \cdot \norm{x_{-k}}_2^2.
\end{equation*}
Then with probability at least $0.99$, it holds that~\smash{$|I_\lambda| \leq |I| / 2^{\lambda-1}$} for all integers $\lambda \geq 1$.
\end{lemma}
\begin{proof}
Recall that in the sketch we pick the parameters $r = 10 \log n$ and $s \geq 8192 c k / \epsilon$ for the estimate $\hat x$ provided by Count-Sketch. Therefore, for any $i \in I$, \cref{lem:count-sketch} implies that:
\begin{equation*}
    \Pr\parens*{|x_i - \hat x_i|^2 \geq \frac{\epsilon (\lambda-1)}{2^{14} k \log n} \cdot \norm{x_{-k}}_2^2} \leq \Pr\parens*{|x_i - \hat x_i|^2 \geq \frac{10 (\lambda-1)}{r} \cdot \frac{\norm{x_{-s}}_2^2}{s}} \leq \exp(-10(\lambda-1)).
\end{equation*}
Therefore, $\Ex[|I_\lambda|] \leq |I| \cdot \exp(-10 (\lambda - 1))$, and from Markov's inequality it follows that
\begin{equation*}
    \Pr\parens*{|I_\lambda| > \frac{|I|}{2^{\lambda-1}}} \leq \frac{\Ex[|I_\lambda|] \cdot 2^{\lambda-1}}{|I|} \leq 2^{\lambda-1} \cdot 2^{-10(\lambda-1)} = 2^{-9(\lambda - 1)}.
\end{equation*}
Taking a union bound over these error probabilities for $\lambda \geq 2$ (noting that for $\lambda = 1$ the statement is clear), the total error probability is at most~\smash{$\sum_{\lambda \geq 2} 2^{-9(\lambda - 1)} \leq 2^{-8} \leq 0.01$} and the proof is complete.
\end{proof}

This completes the analysis of the point and tail estimates.  For the rest of the analysis we will never directly consider the randomness of these estimates, other than conditioning on the previous \cref{lem:point-estimates,lem:tail-estimates,lem:point-estimates-intermediate}.

\subsection{Estimation Error per Bucket} \label{sec:sketch:sec:bucket-estimates}
The next step is to bound the total estimation error in each bucket $B_v$. In fact, we only care about the estimation error of the intermediate coordinates $I_v$ (as these are the only potentially recovered elements that influence the values $p_{v, j}$ in the recovery algorithm; we will be more precise about this in the next subsection).

\begin{lemma}[Bucket Load] \label{lem:bucket-load}
With high probability, it holds that \smash{$|I_v| \leq \frac{2 \log n}{\log\log n}$}.
\end{lemma}
\begin{proof}
Fix any bucket $v \in V$. By \cref{lem:heavy-intermediate-number} the expected bucket load is at most
\begin{equation*}
    \Ex[|I_v|] = |I| \cdot \frac{3}{|V|} \leq 8c k / \epsilon \cdot \frac{3}{24 c k / \epsilon} = 1.
\end{equation*}
Moreover, as across all $i \in I$ the events $i \in I_v$ are independent, from Chernoff's bound it follows for $\ell := \frac{2\log n}{\log\log n}$ that
\begin{equation*}
    \Pr\parens*{|I_v| \geq \ell} \leq \frac{e^{\ell+1}}{\ell^\ell} \leq \frac{n^{o(1)}}{n^{2-o(1)}}.
\end{equation*}
Taking a union bound over all $|V| \leq O(n)$ buckets, the claim is immediate.
\end{proof}

\begin{lemma}[Estimation Error per Bucket] \label{lem:bucket-error}
Conditioned on the event from \cref{lem:point-estimates-intermediate}, with high probability it holds that, for each bucket $v \in V$,
\begin{equation*}
    \norm{(x - \hat x)_{I_v}}_2^2 \leq \frac{\epsilon}{2048 k} \cdot \norm{x_{-k}}_2^2.
\end{equation*}
\end{lemma}
\begin{proof}
We condition on the event from \cref{lem:point-estimates-intermediate} and treat the randomness of the point estimates~$\hat x$ as fixed. The only source of randomness that concerns us here is the randomness of $e_i$ for $i \in I$. Fix any bucket $v \in V$ and suppose that
\begin{equation*}
    \norm{(x - \hat x)_{I_v}}_2^2 > \frac{\epsilon}{2048 k} \cdot \norm{x_{-k}}_2^2.
\end{equation*}
On the other hand, we have that
\begin{align*}
    \norm{(x - \hat x)_{I_v}}_2^2
    &= \sum_{\lambda \geq 1} \norm{(x - \hat x)_{B_v \cap I_\lambda}}_2^2
    \leq \sum_{\lambda \geq 1} |B_v \cap I_\lambda| \cdot \frac{\epsilon\lambda}{2^{14} k \log n} \cdot \norm{x_{-k}}_2^2,
\end{align*}
and thus
\begin{equation*}
    \sum_{\lambda \geq 1} \lambda \cdot |B_v \cap I_\lambda| > 8 \log n.
\end{equation*}
However, we will show that this last inequality holds only with very small probability. To see this, first recall from \cref{lem:point-estimates-intermediate} that $H_\lambda = \emptyset$ for all $\lambda > \log n$. Now fix any integers $b = (b_1, \dots, b_{\log n})$ with~\smash{$\sum_{\lambda \geq 1} \lambda b_\lambda \geq 8 \log n$}; we show that the probability that $b_\lambda = |B_v \cap I_\lambda|$ is small. A first observation is that if $\sum_{\lambda \geq 1} b_{\lambda} > \ell := \frac{2 \log n}{\log \log n}$ then the bucket load $|I_v|$ exceeds $\ell$, but this event happens only with small probability by the previous lemma. So instead assume that~\smash{$\sum_{\lambda \geq 1} b_\lambda \leq \ell$}. Recall that each element $i \in I_\lambda$ is participating in $B_v$ independently with probability $\frac{3}{|V|}$. Therefore:
\begin{align*}
    \Pr(|B_v \cap I_\lambda| = b_\lambda)
    &\leq \binom{|I_\lambda|}{b_\lambda} \cdot \parens*{\frac{3}{|V|}}^{b_\lambda} \\
    &\leq |I_\lambda|^{b_\lambda} \cdot \parens*{\frac{3}{|V|}}^{b_\lambda} \\
    &\leq \frac{|I|^{b_\lambda}}{2^{(\lambda-1) b_{\lambda}}} \cdot \parens*{\frac{3}{|V|}}^{b_\lambda} \\
    &\leq \frac{1}{2^{(\lambda-1) b_\lambda}},
\end{align*}
where in the final step we used that $|I| \leq 8 c k / \epsilon \leq |V| / 3$. These events are independent for the different sets $I_\lambda$, hence
\begin{align*}
    \Pr(\forall \lambda \in [\log n] : |B_v \cap I_\lambda| = b_\lambda)
    &\leq \prod_\lambda \frac{1}{2^{(\lambda-1) b_\lambda}} \\
    &= 2^{-\sum_\lambda (\lambda-1) b_\lambda} \\
    &= 2^{-\sum_\lambda \lambda b_\lambda + \sum_\lambda b_\lambda} \\
    &\leq 2^{-8 \log n + \ell} \\
    &\leq n^{-8+o(1)}.
\end{align*}
To bound the total error probability, we take a union bound over the $|V| \leq O(n)$ buckets $v$, and also over all integer sequences $b = (b_1, \dots, b_{\log n})$ with $\sum_\lambda \lambda b_\lambda \geq 8 \log n$ and $\sum_\lambda b_\lambda \leq \ell$. Each such sequence has at most $\ell$ nonzero entries which respectively ranges from $1$ to $\ell$, and thus the number of sequences is at most
\begin{equation*}
    \binom{\log n}{\ell} \cdot (\ell+1)^{\ell} \leq (\log n)^{\ell} \cdot (\log n)^\ell = n^4.
\end{equation*}
It follows that the error probability is $n^{-3+o(1)}$, which completes the proof.
\end{proof}

\subsection{Model as Weighted Hypergraph Peeling} \label{sec:sketch:sec:peeling}
We are finally ready to cast the recovery algorithm as a weighted hypergraph peeling process. Specifically, consider the following random hypergraph:

\begin{definition}[Associated Hypergraph] \label{def:associated-hypergraph}
The \emph{associated hypergraph} $G = (V, E, w)$ is the $3$\=/uniform hypergraph with 
\begin{itemize}[topsep=\medskipamount, itemsep=0pt]
    \item vertices $V = [2^{19} c k/\epsilon]$,
    \item hyperedges $E = \set{ e_i : i \in H}$,
    \item edge-weights $w_{e_i} = |x_i|^2$, and
    \item vertex-weights $w_v = \norm{x_{L_v}}_2^2$.
\end{itemize}
\end{definition}

In the remainder of this subsection we will prove that, with good probability, \emph{all} peelable edges~$e_i$ in the associated hypergraph will be recovered by the algorithm in the sense that $i \in R$. In the next subsection we will then prove that for almost all coordinates $i$ that we aim to recover the edge $e_i$ is peelable in the associated hypergraph.

\begin{lemma}[Recovering Free Edges] \label{lem:recover-free}
Suppose that the events from \cref{lem:point-estimates,lem:tail-estimates,lem:bucket-load,lem:bucket-error} hold. Then with high probability, the graph $G \setminus \set{e_i : i \in R}$ does not contain $\rho$-free edges for $\rho = 2048$.
\end{lemma}
\begin{proof}
Suppose otherwise, and let $e_{i^*}$ be a $\rho$-free edge in the graph $G \setminus \set{e_i : i \in R}$. Let $v \in e_{i^*}$ be the vertex witnessing that $e_{i^*}$ is $\rho$-free. Focus on the last iteration (of the while-loop in \crefrange*{alg:recover:line:loop}{alg:recover:line:push}) where we have considered the vertex~$v$. As this is the last iteration considering~$v$, from this point on the algorithm cannot recover any more edges $e_i$ incident to $v$ (as this would cause the algorithm to reinsert $v$ into the queue~$Q$), and in particular the set $B_v \cap R$ has not changed from this iteration until the algorithm has terminated. We show that this is a contradiction, since the algorithm would have recovered~\makebox{$i = i^*$} in \cref*{alg:recover:line:dec} and inserted $i$ into $R$ in \cref*{alg:recover:line:push} in this iteration.

By definition, as $e_i$ is $\rho$-free in $G \setminus \set{e_i : i \in R}$ we have that
\begin{equation*}
    w_{e_{i^*}} \geq \rho \cdot \left(w_v + \sum_{\substack{i \in H_v \setminus R\\i \neq i^*}} w_{e_i}\right).
\end{equation*}
Plugging in the definitions of the vertex- and edge-weights this becomes
\begin{equation*}
    |x_{i^*}|^2 \geq \rho \cdot \left(\sum_{i \in L_v} |x_i|^2 + \sum_{\substack{i \in H_v \setminus R\\i \neq i^*}} |x_i|^2 \right) \geq \rho \cdot \sum_{\substack{i \in B_v \setminus R\\i \neq i^*}} |x_i|^2.
\end{equation*}
By definition the algorithm computes
\begin{align*}
    p_{v, j}
    &= \sum_{i \in B_v} \sigma_j(i) \cdot \operatorname{Enc}(i)_j \cdot x_i - \sum_{i \in B_v \cap R} \sigma_j(i) \cdot \operatorname{Enc}(i)_j \cdot \hat x_i \\
    &= \sum_{i \in B_v \setminus R} \sigma_j(i) \cdot \operatorname{Enc}(i)_j \cdot x_i + \sum_{i \in B_v \cap R} \sigma_j(i) \cdot \operatorname{Enc}(i)_j \cdot (x_i - \hat x_i),
\end{align*}
and we define
\begin{align*}
    p_{v, j}^*
    &= \vphantom{\sum_i} p_{v, j} - \sigma_j(i^*) \cdot \operatorname{Enc}(i^*)_j \cdot x_{i^*} \\
    &= \sum_{\substack{i \in B_v \setminus R\\i \neq i^*}} \sigma_j(i) \cdot \operatorname{Enc}(i)_j \cdot x_i + \sum_{i \in B_v \cap R} \sigma_j(i) \cdot \operatorname{Enc}(i)_j \cdot (x_i - \hat x_i).
\end{align*}
We view $p_{v, j}^*$ as a random variable over the randomness of the signs $\sigma_j$. Of course, the set $R$ might also depend on the randomness of the $\sigma_j$'s, however, for now we suppose that $R$ is fixed. We justify this assumption later. We clearly have that
\begin{equation*}
    \Ex[p_{v, j}^*] = 0
\end{equation*}
and
\begin{equation*}
    \Var[p_{v, j}^*]
    = \sum_{\substack{i \in B_v \setminus R\\i \neq i^*}} \operatorname{Enc}(i)_j \cdot |x_i|^2 + \sum_{i \in B_v \cap R} \operatorname{Enc}(i)_j \cdot |x_i - \hat x_i|^2.
\end{equation*}
To derive an upper bound on this variance, we bound the first sum by \smash{$\frac{|x_{i^*}|^2}{\rho}$} due to our previous calculation. To bound the second sum, first recall that in the events of \cref{lem:point-estimates,lem:tail-estimates}, \cref{lem:recovery-threshold} implies that $R \subseteq I$ and thus $B_v \cap R \subseteq I_v$. Second, $\norm{x_{I_v}}_2^2$ is upper bounded by \cref{lem:bucket-error}. Putting both together, we obtain that
\begin{equation*}
    \Var[p_{v, j}^*] \leq \frac{|x_{i^*}|^2}{\rho} + \frac{\epsilon}{\rho k} \cdot \norm{x_{-k}}_2^2 \leq \frac{2|x_{i^*}|^2}{\rho},
\end{equation*}
where in the last step we used that $i^*$ is heavy.
From Chebyshev's inequality it follows that
\begin{equation*}
    \Pr\parens*{|p_{v, j}^*| \geq \frac{|x_{i^*}|}{4}} \leq \frac{\Var[p_{v, j}^*]}{(\frac{|x_{i^*}|}{4})^2} \leq \frac{32}{\rho} = \frac{1}{64}.
\end{equation*}
Let us call an index $j$ \emph{bad} if~\smash{$|p_{v, j}^*| \geq \frac{|x_{i^*}|}{4}$} and \emph{good} otherwise. The previous computation shows that we expect at most a $\frac{1}{64}$-fraction of the indices $j \in [2048 \log n]$ to be bad, and by Chernoff's bound at most a $\frac{1}{32}$-fraction is bad with high probability (namely, \smash{$1 - \exp(-\frac{32 \log n}{3}) \geq 1 - n^{-10}$}).

Now focus on any good index~$j$. In case that $\operatorname{Enc}(i^*)_j = 1$ we have that
\begin{equation*}
    |p_{v, j}| \geq |x_{i^*}| - \frac{|x_{i^*}|}{4} = \frac{3}{4} \cdot |x_{i^*}|,
\end{equation*}
whereas if $\operatorname{Enc}(i^*)_j = 0$ then
\begin{equation*}
    |p_{v, j}| = |p_{v, j}^*| \leq \frac{1}{4} \cdot |x_{i^*}|.
\end{equation*}
Recall that the bit-string $y$ is obtained from $|p_{v, j}|$ by replacing the half of the entries with largest magnitudes by $1$ and the others by $0$. From the previous paragraph, and from the assumption that the fraction of $1$'s in $\operatorname{Enc}(i^*)$ is $\frac{1}{2} \pm \frac{1}{16}$, it follows that $y$ differs from $\operatorname{Enc}(i^*)$ in at most a fraction of~\smash{$\frac{1}{16} + \frac{1}{16} = \frac{1}{8}$} of the entries. (Indeed, in the event that all indices are good, $y$ differs in only a fraction of $\frac{1}{16}$ entries from $\operatorname{Enc}(i^*)$. And each bad coordinate contributes one mismatch plus possibly one displaced good coordinate, leading to a fraction of at most $2 \cdot \frac{1}{32}$ mismatches.) By \cref{lem:code} it finally follows that $\operatorname{Dec}(y) = i^*$.

Additionally, \cref{lem:recovery-threshold} guarantees that (conditioned on the events from \cref{lem:point-estimates,lem:tail-estimates}) every heavy index $i$ satisfies that $|x_i|^2 \geq \frac{\epsilon}{2k} \cdot \hat t$. In particular, $i^*$ passes the if-condition in \cref*{alg:recover:line:threshold} and we successfully insert $i^*$ into $R$ in \cref*{alg:recover:line:push}.

This shows that for any bucket $v$ \emph{and for a fixed set $R$} the error probability is at most $n^{-10}$. However, the set $R$ depends also on the randomness as in previous iterations the signs $\sigma_j$ have influenced which heavy coordinates could be recovered. To deal with this dependency issue, note that the argument from before in fact only depends on $R \cap B_v$. However, conditioning on the events from \cref{lem:point-estimates,lem:tail-estimates} we throughout have that $R \subseteq I$ by  \cref{lem:recovery-threshold}, and conditioning on the event from \cref{lem:bucket-load} we have that $|I_v| \leq O(\log n / \log\log n)$. Notably, both of these events are independent of the randomness of the $\sigma_j$'s. Thus, we can afford to take a union bound over \emph{all possible sets $R \cap B_v$}---namely, the $2^{O(\log n / \log\log n)} = n^{o(1)}$ subsets of $I_v$---and over the at most $O(k/\epsilon) = O(n)$ buckets $v$. The total error probability is $n^{-9+o(1)}$.
\end{proof}

\begin{lemma}[Recovering Peelable Edges] \label{lem:recover-peelable}
Suppose that the events from \cref{lem:point-estimates,lem:tail-estimates,lem:bucket-load,lem:bucket-error} hold. If an edge $e_i$ is peelable in the associated hypergraph $G$, then with high probability we recover~$i$ (i.e., $i \in R$).
\end{lemma}
\begin{proof}
If $e_i$ is $\rho$-peelable in $G$ then by definition there is a peeling sequence $e_{i_1}, \dots, e_{i_\ell}$ with $i = i_{\ell}$. Seeking contradiction, assume that $i = i_\ell \not\in R$. Let $j \leq \ell$ be the smallest index such that $i_j \not\in R$. As the edge $e_{i_j}$ is $\rho$-free in the graph $G \setminus \set{e_{i_1}, \dots, e_{i_{j-1}}}$, it is thus also $\rho$-free in $G \setminus \set{e_i : i \in R}$. This contradicts \cref{lem:recover-free}.
\end{proof}

\subsection{The Associated Hypergraph is Well-Peelable} \label{sec:sketch:sec:peelable}
In the previous step we have established that all edges that are peelable in the associated hypergraph can be recovered by the algorithm. In this section we show that almost all edges $e_i$ are peelable. More specifically, we will show that the total weight of all edges $e_i$ for $i \in T$ that are not peelable is negligibly small:

\begin{lemma}[Weight of Non-Peelable Edges] \label{lem:non-peelable-weight}
Let $P \subseteq H$ denote the set of $\rho$-peelable edges $e_i$ in the associated hypergraph. Then with probability at least $0.97$, it holds that
\begin{equation*}
    \norm{x_{T \setminus P}}_2^2 \leq 800 \epsilon \cdot \norm{x_{-k}}_2^2.
\end{equation*}
\end{lemma}
\begin{proof}
Observe that the associated hypergraph relies on two sources of randomness: $\set{ e_i : i \in H}$ and $\set{ e_i : i \in L}$. The former (the \emph{heavy} randomness) determines the edges $E$ of the associated hypergraph, whereas the latter (the \emph{light} randomness) determines the vertex weights $w_v$.

For the first part of the proof disregard the light randomness and focus on the heavy randomness. On the one hand, \smash{$|E| = |H| \leq 2 k / \epsilon \leq \frac{|V|}{6}$} by \cref{lem:heavy-intermediate-number} and thus \cref{lem:hypertrees-unicyclic} states that with high probability each connected component in the associated hypergraph is a hypertree or unicyclic. On the other hand, \smash{$|E| \leq 2k \epsilon \leq \frac{|V|}{72 \rho}$}, and thus \cref{lem:expected-spreadness} implies that $\Ex[D(e_i)] \leq \rho + 1$ for each~\makebox{$i \in H$}. It follows that
\begin{equation*}
    \Ex\brackets*{\sum_{i \in T} D(e_i)} \leq (\rho + 1) k \leq 2 \rho k,
\end{equation*}
and by Markov's bound with probability at least $0.99$ it holds that
\begin{equation*}
    \sum_{i \in T} D(e_i) \leq 200\rho k.
\end{equation*}

Next, fix the heavy randomness and focus on the light randomness. For each vertex $v \in V$, the expected vertex-weight is at most
\begin{align*}
    \Ex[w_v]
    &= \sum_{i \in L} \Pr(v \in e_i) \cdot |x_i|^2 \\
    &\leq \frac{3}{|V|} \sum_{i \in L} |x_i|^2 \\
    &= \frac{3}{|V|} \cdot \parens*{\norm{x_{L \setminus T}}_2^2 + \norm{x_{L \cap T}}_2^2} \\
    &\leq \frac{3}{|V|} \cdot \parens*{\norm{x_{-k}}_2^2 + k \cdot \frac{\epsilon}{k} \cdot \norm{x_{-k}}_2^2} \\
    &\leq \frac{6}{|V|} \cdot \norm{x_{-k}}_2^2 \\
    &\leq \frac{\epsilon}{2^{16} k} \cdot \norm{x_{-k}}_2^2.
\end{align*}
In this setting---where the hypergraph and the edge weights are fix, each connected component is a hypertree or unicyclic, and the vertex weights are chosen randomly (though not necessarily independently) with expected weight at most $\mu = \frac{\epsilon}{2^{16}k} \cdot \norm{x_{-k}}_2^2$---\cref{lem:peeling-random-vertex-weights} yields that the probability that an edge $e_i$ is not $\rho$-peelable is at most
\begin{equation*}
    \Pr(i \not\in P) \leq \frac{\mu \cdot D(e_i)}{w_{e_i}} = \frac{\epsilon \cdot D(e_i)}{2^{16} k \cdot |x_i|^2} \cdot \norm{x_{-k}}_2^2.
\end{equation*}
Thus, the expected weight of the non-peelable edges is
\begin{align*}
    \Ex\brackets*{\norm{x_{T \setminus P}}_2^2}
    &= \sum_{i \in T} \Pr(i \not\in P) \cdot |x_i|^2 \\
    &\leq \sum_{i \in T} \frac{\epsilon \cdot D(e_i)}{2^{16} k} \cdot \norm{x_{-k}}_2^2 \\
    &\leq \frac{\epsilon \cdot 200 \rho k}{2^{16} k} \cdot \norm{x_{-k}}_2^2 \\
    &\leq \vphantom{\frac{\epsilon}{\epsilon}}8 \epsilon \norm{x_{-k}}_2^2.
\end{align*}
By one more application of Markov's inequality, we finally obtain that $\norm{x_{T \setminus P}}_2^2 \leq 800 \epsilon \norm{x_{-k}}_2^2$ with probability at least $0.99$. Taking a union bound over the all four error events---the two events that succeed with high probability,\footnote{Strictly speaking, these events happen only with high probability $1 - |V|^{-\Omega(1)}$ in the number of nodes of the associated hypergraph, $|V| = 2^{19} c k / \epsilon$. If $k$ and $\epsilon$ are both constants, then this failure probability is also only constant. However, if necessary we can further increase the number of buckets (by a constant factor) to achieve failure probability $0.01$.} plus two applications of Markov's inequality---the total error probability is indeed $o(1) + 0.01 + 0.01 \leq 0.03$. 
\end{proof}

The following lemma summarizes the analysis up to this point:

\begin{lemma}[Recovery] \label{lem:recover}
Conditioned on the events from \cref{lem:recover-peelable,lem:non-peelable-weight}, the algorithm recovers a set $R$ satisfying that
\begin{equation*}
    \norm{x_{T \setminus R}}_2^2 \leq 800 \epsilon \cdot \norm{x_{-k}}_2^2.
\end{equation*}
\end{lemma}
\begin{proof}
Conditioned on the event from \cref{lem:non-peelable-weight}, the set $P \subseteq H$ of $\rho$-peelable edges in the associated hypergraph satisfies that $\norm{x_{T \setminus P}}_2^2 \leq 800 \epsilon \cdot \norm{x_{-k}}_2^2$. Conditioned on the event from \cref{lem:recover-peelable} the set $R$ computed by the algorithm contains all $\rho$-peelable edges in $G$. The claim follows.
\end{proof}

\subsection{Refinement} \label{sec:sketch:sec:correctness}
In the previous steps of the analysis we have shown that the algorithm successfully recovers a set~$R$ such that the total weight of the top-$k$ coordinates missing in $R$ is small. We finally analyze that the vector $x'$ as computed by the recovery algorithm is as claimed. Recall that the recovery algorithm first selects $S \subseteq R$ to be the set of the $\min(3k, |R|)$ coordinates $i$ with largest estimates~$|\hat x_i|$, and then selects $x' = x_S$. The proof is by an exchange argument that is exactly analogous to~\cite{NakosS19} and similar also to~\cite{PriceW11}.

\begin{lemma}[Refinement] \label{lem:refinement}
Conditioned on the events from \cref{lem:point-estimates,lem:recover-peelable,lem:non-peelable-weight}, the algorithm returns a vector $x'$ satisfying that
\begin{equation*}
    \norm{x - x'}_2^2 \leq (1 + 900 \epsilon) \norm{x_{-k}}_2^2.
\end{equation*}
\end{lemma}
\begin{proof}
For a set $A \subseteq [n]$ we write $\bar A = [n] \setminus A$ to denote its complement set. By construction we have that:
\begin{align*}
    \norm{x - x'}_2^2
    &= \norm{(x - x')_S}_2^2 + \norm{(x - x')_{\bar S}}_2^2 \\
    &= \norm{(x - \hat x)_S}_2^2 + \norm{x_{\bar S}}_2^2 \\
    &\leq \frac{|S| \cdot \epsilon}{k} \cdot \norm{x_{-k}}_2^2 + \norm{x_{\bar S}}_2^2 \\
    &= \epsilon \cdot \norm{x_{-k}}_2^2 + \norm{x_{\bar S}}_2^2,
\end{align*}
and thus it suffices to bound $\norm{x_{\bar S}}_2^2$ in the following. First note that if $|S| < 3k$, then $S = R$ and the claim follows immediately by
\begin{align*}
    \norm{x_{\bar S}}_2^2
    &= \norm{x_{\bar R}}_2^2 \\
    &= \norm{x_{\bar R \cap T}}_2^2 + \norm{x_{\bar R \cap \bar T}}_2^2 \\
    &\leq 800 \epsilon \cdot \norm{x_{-k}}_2^2 + \norm{x_{-k}}_2^2.
\end{align*}
So assume that $|S| = 3k$. We partition $\bar S = (\bar S \cap R \cap T) \sqcup (\bar R \cap T) \sqcup (\bar S \cap \bar T)$ and bound
\begin{align*}
    \norm{x_{\bar S}}_2^2
    &= \norm{x_{\bar S \cap R \cap T}}_2^2 + \norm{x_{\bar R \cap T}}_2^2 + \norm{x_{\bar S \cap \bar T}}_2^2 \\
    &= \norm{x_{\bar S \cap R \cap T}}_2^2 + \norm{x_{\bar R \cap T}}_2^2 + \norm{x_{\bar T}}_2^2 - \norm{x_{S \cap \bar T}}_2^2 \\
    &\leq 800 \epsilon \cdot \norm{x_{-k}}_2^2 + \norm{x_{-k}}_2^2 + \norm{x_{\bar S \cap R \cap T}}_2^2 - \norm{x_{S \cap \bar T}}_2^2,
\end{align*}
where in the last step we used that $\norm{x_{\bar R \cap T}}_2^2 \leq 800 \epsilon \cdot \norm{x_{-k}}_2^2$ by \cref{lem:recover} and $\norm{x_{\bar T}}_2^2 = \norm{x_{-k}}_2^2$ by definition. In what follows we show that $\norm{x_{\bar S \cap R \cap T}}_2^2 - \norm{x_{S \cap \bar T}}_2^2 \leq 8\epsilon \cdot \norm{x_{-k}}_2^2$.

Let $i \in \bar S \cap R \cap T$ be such that $|x_i|$ is maximized, and let $j \in S \cap \bar T$ be such that $|x_j|$ is minimized. On the one hand, clearly~\smash{$|x_i| \geq |x_j|$} since $i$ is among the $k$ largest coordinates in $x$, but $j$ is not. On the other hand, we have that $|\hat x_i| \leq |\hat x_j|$, as otherwise we would have included $i$ in place of $j$ into~$S$. It follows that
\begin{equation*}
    |x_i| \leq |\hat x_i| + \sqrt{\frac{\epsilon}{k}} \cdot \norm{x_{-k}}_2 \leq |\hat x_j| + \sqrt{\frac{\epsilon}{k}} \cdot \norm{x_{-k}}_2 \leq |x_j| + \sqrt{\frac{4\epsilon}{k}} \cdot \norm{x_{-k}}_2,
\end{equation*}
where have applied the fact that $\hat x$ approximates $x$ with coordinate-wise error $\frac{\epsilon}{k} \cdot \norm{x_{-k}}_2$ twice. Combining both bounds we conclude that~\smash{$|\, |x_i| - |x_j| \,|^2 \leq \frac{4\epsilon}{k} \cdot \norm{x_{-k}}_2^2$}. With this in mind, observing that $|S \cap \bar T| \leq |S| - |T| = 2k$, and applying the inequality $a^2 - 2b^2 \leq 2(a - b)^2$,\footnote{Indeed, for all $a, b \in \Real$ we have $2(a - b)^2 - (a^2 - 2b^2) = a^2 - 4ab + 4b^2 = (a - 2b)^2 \geq 0$, and thus $a^2 - 2b^2 \leq 2(a - b)^2$.} we finally have that:
\begin{align*}
    \norm{x_{\bar S \cap R \cap T}}_2^2 - \norm{x_{S \cap \bar T}}_2^2
    &\leq k \cdot |x_i|^2 - 2k \cdot |x_j|^2 \\
    &\leq 2k \cdot (|x_i| - |x_j|)^2 \\
    &\leq 2k \cdot \frac{4\epsilon}{k} \cdot \norm{x_{-k}}_2^2 \\
    &= 8 \epsilon \cdot \norm{x_{-k}}_2^2.
\end{align*}
The proof is complete by putting together all previous bounds.
\end{proof}

\begin{lemma}[Correctness] \label{lem:correctness}
With probability at least $0.95$, the algorithm returns a vector $x'$ satisfying that
\begin{equation*}
    \norm{x - x'}_2^2 \leq (1 + 900 \epsilon) \norm{x_{-k}}_2^2.
\end{equation*}
\end{lemma}
\begin{proof}
The events from \cref{lem:point-estimates,lem:tail-estimates,lem:bucket-load} each happen with high probability $1 - o(1)$. The event from \cref{lem:point-estimates-intermediate} happens with probability $0.99$, and in this case the event from \cref{lem:bucket-error} happens with high probability $1 - o(1)$. In this case further the event of \cref{lem:recover-peelable} happens with high probability $1 - o(1)$, and finally the event from \cref{lem:non-peelable-weight} happens with probability $0.97$. Conditioned on these events, \cref{lem:refinement} implies the correctness claim. The total error probability is $0.1 + 0.3 + o(1) \leq 0.05$.
\end{proof}

\begin{lemma}[\#Rows and Column Sparsity] \label{lem:column-sparsity}
The sketch matrix $A$ has $m = O((k / \epsilon) \log n)$ rows and column sparsity $O(\log n)$.
\end{lemma}
\begin{proof}
There are three components to our sketch: First, the tail estimation sketch has $O(\log n)$ rows and column sparsity $O(\log n)$ (\cref{lem:tail-sketch}). Second, Count-Sketch has $O(s r) = O((k / \epsilon) \log n)$ rows and column sparsity $O(r) = O(\log n)$ (\cref{lem:count-sketch}). As the third component we have one row for each query $q_{v, j}$. As there are $O(k / \epsilon)$ buckets $v$ and $O(\log n)$ repetitions $j$, this amounts to $O((k / \epsilon) \log n)$ rows. Moreover, each coordinate $i$ appears only with nonzero coefficient in the rows~$q_{v, j}$ with $v \in e_i$, so the column sparsity is $O(\log n)$.
\end{proof}

\subsection{Running Time} \label{sec:sketch:sec:time}
The final step is to show that the recovery algorithm runs in time $O((k / \epsilon) \log n)$.

\begin{lemma}[Running Time] \label{lem:time}
The recovery algorithm runs in time $O((k / \epsilon) \log n)$ with high probability.
\end{lemma}
\begin{proof}
The first step is to bound the number of iterations of the while-loop in \crefrange*{alg:recover:line:loop}{alg:recover:line:push} by $O(|V| + |R|)$. Let us call an iteration \emph{successful} if the condition in \cref*{alg:recover:line:threshold} is not fulfilled. Observe that in each successful iteration we increase $|Q|$ by $2$ and increase $|R|$ by $1$, whereas in each unsuccessful iteration we decrease $|Q|$ by $1$ and leave $R$ unchanged. Thus, there can clearly be at most $|R|$ successful iterations, and consequently at most $|V| + 2|R|$ unsuccessful iterations.

Next, we bound the running time of a single iteration. The time to pick $v \in Q$ in \cref*{alg:recover:line:pop} is negligible. Let us ignore \cref*{alg:recover:line:residual} for a moment. Computing the string $y$ in \cref*{alg:recover:line:median} essentially amounts to computing the median of the values $p_{v, j}$ which runs in linear time $O(\log n)$. Decoding~$y$ in \cref*{alg:recover:line:dec} also runs in time $O(\log n)$ by \cref{lem:code}. \cref*{alg:recover:line:threshold,alg:recover:line:push} are again negligibly fast. Thus, disregarding the computation of the values $p_{v, j}$ in \cref*{alg:recover:line:residual}, the time per iteration is $O(\log n)$. To amount for \cref*{alg:recover:line:residual}, we can simply keep track of the values $p_{v, j}$ throughout. Whenever we insert an element $i$ into $R$ we compute $\operatorname{Enc}(i)$ in time $O(\log n)$ and update the values $p_{v, j}$ for the three buckets $v \in e_i$ in time $O(\log n)$. It follows that the total time is $O((|V| + |R|) \log n)$.

Finally, recall that $|V| = O(k / \epsilon)$. By \cref{lem:recovery-threshold} we have that $R \subseteq I$ with high probability, and thus $|R| \leq |I| \leq O(k / \epsilon)$ by \cref{lem:heavy-intermediate-number}.
\end{proof}

This finally completes the proof of \cref{thm:main}:

\begin{proof}[Proof of \cref{thm:main}.]
In \cref{lem:column-sparsity} we have demonstrated that the number of rows and the column sparsity of the sketch matrix $A$ are as claimed. In \cref{lem:correctness} we have shown that the recovery algorithm recovers an approximation $x'$ satisfying that $\norm{x - x'}_2^2 \leq (1 + 900 \epsilon) \norm{x_{-k}}_2^2$ with constant probability $0.9$; thus, by rescaling $\epsilon$ by a constant we achieve the claimed bound. In \cref{lem:time} we have shown that the recovery algorithm runs in time $O((k / \epsilon) \log n)$.
\end{proof}

\section*{Acknowledgements}
We thank the anonymous reviewers for many helpful comments on an earlier draft of this paper.

\bibliographystyle{plainurl}
\bibliography{main}

\begin{thebibliography}{10}

\bibitem{CandesRT06}
Emmanuel~J. Cand{\`{e}}s, Justin~K. Romberg, and Terence Tao.
\newblock Robust uncertainty principles: exact signal reconstruction from highly incomplete frequency information.
\newblock {\em {IEEE} Trans. Inf. Theory}, 52(2):489--509, 2006.
\newblock \href {https://doi.org/10.1109/TIT.2005.862083} {\path{doi:10.1109/TIT.2005.862083}}.

\bibitem{CandesT05}
Emmanuel~J. Cand{\`{e}}s and Terence Tao.
\newblock Decoding by linear programming.
\newblock {\em {IEEE} Trans. Inf. Theory}, 51(12):4203--4215, 2005.
\newblock \href {https://doi.org/10.1109/TIT.2005.858979} {\path{doi:10.1109/TIT.2005.858979}}.

\bibitem{CandesT06}
Emmanuel~J. Cand{\`{e}}s and Terence Tao.
\newblock Near-optimal signal recovery from random projections: Universal encoding strategies?
\newblock {\em {IEEE} Trans. Inf. Theory}, 52(12):5406--5425, 2006.
\newblock \href {https://doi.org/10.1109/TIT.2006.885507} {\path{doi:10.1109/TIT.2006.885507}}.

\bibitem{CevherKSZ17}
Volkan Cevher, Michael Kapralov, Jonathan Scarlett, and Amir Zandieh.
\newblock An adaptive sublinear-time block sparse fourier transform.
\newblock In Hamed Hatami, Pierre McKenzie, and Valerie King, editors, {\em 49th Annual {ACM} Symposium on Theory of Computing ({STOC} 2017)}, pages 702--715. {ACM}, 2017.
\newblock \href {https://doi.org/10.1145/3055399.3055462} {\path{doi:10.1145/3055399.3055462}}.

\bibitem{CharikarCF04}
Moses Charikar, Kevin~C. Chen, and Martin Farach{-}Colton.
\newblock Finding frequent items in data streams.
\newblock {\em Theor. Comput. Sci.}, 312(1):3--15, 2004.
\newblock \href {https://doi.org/10.1016/S0304-3975(03)00400-6} {\path{doi:10.1016/S0304-3975(03)00400-6}}.

\bibitem{CheraghchiN20}
Mahdi Cheraghchi and Vasileios Nakos.
\newblock Combinatorial group testing and sparse recovery schemes with near-optimal decoding time.
\newblock In Sandy Irani, editor, {\em 61st Annual {IEEE} Symposium on Foundations of Computer Science ({FOCS} 2020)}, pages 1203--1213. {IEEE}, 2020.
\newblock \href {https://doi.org/10.1109/FOCS46700.2020.00115} {\path{doi:10.1109/FOCS46700.2020.00115}}.

\bibitem{Cohen2008}
Albert Cohen, Wolfgang Dahmen, and Ronald DeVore.
\newblock Compressed sensing and best {$k$}-term approximation.
\newblock {\em Journal of the American Mathematical Society}, 22(1):211--231, 2008.
\newblock URL: \url{http://dx.doi.org/10.1090/S0894-0347-08-00610-3}, \href {https://doi.org/10.1090/s0894-0347-08-00610-3} {\path{doi:10.1090/s0894-0347-08-00610-3}}.

\bibitem{CormodeM06}
Graham Cormode and S.~Muthukrishnan.
\newblock Combinatorial algorithms for compressed sensing.
\newblock In Paola Flocchini and Leszek Gasieniec, editors, {\em 13th International Colloquium on Structural Information and Communication Complexity ({SIROCCO} 2006)}, volume 4056 of {\em Lecture Notes in Computer Science}, pages 280--294. Springer, 2006.
\newblock \href {https://doi.org/10.1007/11780823\_22} {\path{doi:10.1007/11780823\_22}}.

\bibitem{DietzfelbingerW07}
Martin Dietzfelbinger and Christoph Weidling.
\newblock Balanced allocation and dictionaries with tightly packed constant size bins.
\newblock {\em Theor. Comput. Sci.}, 380(1-2):47--68, 2007.
\newblock URL: \url{https://doi.org/10.1016/j.tcs.2007.02.054}, \href {https://doi.org/10.1016/J.TCS.2007.02.054} {\path{doi:10.1016/J.TCS.2007.02.054}}.

\bibitem{DuH99}
Ding-Zhu Du and Frank~Kwang ming Hwang.
\newblock {\em Combinatorial group testing and its applications}, volume~12.
\newblock World Scientific, 1999.

\bibitem{EldarK12}
Yonina~C. Eldar and Gitta Kutyniok.
\newblock {\em Compressed sensing: theory and applications}.
\newblock Cambridge University Press, 2012.

\bibitem{EppsteinG07}
David Eppstein and Michael~T. Goodrich.
\newblock Space-efficient straggler identification in round-trip data streams via newton's identities and invertible bloom filters.
\newblock In Frank K. H.~A. Dehne, J{\"{o}}rg{-}R{\"{u}}diger Sack, and Norbert Zeh, editors, {\em 10th International Workshop on Algorithms and Data Structure ({WADS} 2007)}, volume 4619 of {\em Lecture Notes in Computer Science}, pages 637--648. Springer, 2007.
\newblock \href {https://doi.org/10.1007/978-3-540-73951-7\_55} {\path{doi:10.1007/978-3-540-73951-7\_55}}.

\bibitem{FountoulakisP12}
Nikolaos Fountoulakis and Konstantinos Panagiotou.
\newblock Sharp load thresholds for cuckoo hashing.
\newblock {\em Random Struct. Algorithms}, 41(3):306--333, 2012.
\newblock URL: \url{https://doi.org/10.1002/rsa.20426}, \href {https://doi.org/10.1002/RSA.20426} {\path{doi:10.1002/RSA.20426}}.

\bibitem{FriezeK15}
Alan Frieze and Micha{\l} Karo{\'n}ski.
\newblock {\em Introduction to random graphs}.
\newblock Cambridge University Press, 2015.

\bibitem{FriezeM12}
Alan~M. Frieze and P{\'{a}}ll Melsted.
\newblock Maximum matchings in random bipartite graphs and the space utilization of cuckoo hash tables.
\newblock {\em Random Struct. Algorithms}, 41(3):334--364, 2012.
\newblock URL: \url{https://doi.org/10.1002/rsa.20427}, \href {https://doi.org/10.1002/RSA.20427} {\path{doi:10.1002/RSA.20427}}.

\bibitem{GilbertI10}
Anna~C. Gilbert and Piotr Indyk.
\newblock Sparse recovery using sparse matrices.
\newblock {\em Proc. {IEEE}}, 98(6):937--947, 2010.
\newblock \href {https://doi.org/10.1109/JPROC.2010.2045092} {\path{doi:10.1109/JPROC.2010.2045092}}.

\bibitem{GilbertLPS12}
Anna~C. Gilbert, Yi~Li, Ely Porat, and Martin~J. Strauss.
\newblock Approximate sparse recovery: Optimizing time and measurements.
\newblock {\em {SIAM} J. Comput.}, 41(2):436--453, 2012.
\newblock \href {https://doi.org/10.1137/100816705} {\path{doi:10.1137/100816705}}.

\bibitem{GilbertLPS17}
Anna~C. Gilbert, Yi~Li, Ely Porat, and Martin~J. Strauss.
\newblock For-all sparse recovery in near-optimal time.
\newblock {\em {ACM} Trans. Algorithms}, 13(3):32:1--32:26, 2017.
\newblock \href {https://doi.org/10.1145/3039872} {\path{doi:10.1145/3039872}}.

\bibitem{GilbertSTV07}
Anna~C. Gilbert, Martin~J. Strauss, Joel~A. Tropp, and Roman Vershynin.
\newblock One sketch for all: fast algorithms for compressed sensing.
\newblock In David~S. Johnson and Uriel Feige, editors, {\em 39th Annual {ACM} Symposium on Theory of Computing ({STOC} 2007)}, pages 237--246. {ACM}, 2007.
\newblock \href {https://doi.org/10.1145/1250790.1250824} {\path{doi:10.1145/1250790.1250824}}.

\bibitem{GoodrichM11}
Michael~T. Goodrich and Michael Mitzenmacher.
\newblock Invertible bloom lookup tables.
\newblock In {\em 49th Annual Allerton Conference on Communication, Control, and Computing}, pages 792--799. {IEEE}, 2011.
\newblock URL: \url{https://doi.org/10.1109/Allerton.2011.6120248}, \href {https://doi.org/10.1109/ALLERTON.2011.6120248} {\path{doi:10.1109/ALLERTON.2011.6120248}}.

\bibitem{GuruswamiI05}
Venkatesan Guruswami and Piotr Indyk.
\newblock Linear-time encodable/decodable codes with near-optimal rate.
\newblock {\em {IEEE} Trans. Inf. Theory}, 51(10):3393--3400, 2005.
\newblock \href {https://doi.org/10.1109/TIT.2005.855587} {\path{doi:10.1109/TIT.2005.855587}}.

\bibitem{HassaniehIKP12b}
Haitham Hassanieh, Piotr Indyk, Dina Katabi, and Eric Price.
\newblock Nearly optimal sparse fourier transform.
\newblock In Howard~J. Karloff and Toniann Pitassi, editors, {\em 44th Annual {ACM} Symposium on Theory of Computing ({STOC} 2012)}, pages 563--578. {ACM}, 2012.
\newblock \href {https://doi.org/10.1145/2213977.2214029} {\path{doi:10.1145/2213977.2214029}}.

\bibitem{HassaniehIKP12a}
Haitham Hassanieh, Piotr Indyk, Dina Katabi, and Eric Price.
\newblock Simple and practical algorithm for sparse fourier transform.
\newblock In Yuval Rabani, editor, {\em 23rd Annual {ACM-SIAM} Symposium on Discrete Algorithms ({SODA} 2012)}, pages 1183--1194. {SIAM}, 2012.
\newblock \href {https://doi.org/10.1137/1.9781611973099.93} {\path{doi:10.1137/1.9781611973099.93}}.

\bibitem{HermanS09}
Matthew~A. Herman and Thomas Strohmer.
\newblock High-resolution radar via compressed sensing.
\newblock {\em {IEEE} Trans. Signal Process.}, 57(6):2275--2284, 2009.
\newblock \href {https://doi.org/10.1109/TSP.2009.2014277} {\path{doi:10.1109/TSP.2009.2014277}}.

\bibitem{IndykK14}
Piotr Indyk and Michael Kapralov.
\newblock Sample-optimal fourier sampling in any constant dimension.
\newblock In {\em 55th Annual {IEEE} Symposium on Foundations of Computer Science ({FOCS} 2014)}, pages 514--523. {IEEE} Computer Society, 2014.
\newblock \href {https://doi.org/10.1109/FOCS.2014.61} {\path{doi:10.1109/FOCS.2014.61}}.

\bibitem{IndykKP14}
Piotr Indyk, Michael Kapralov, and Eric Price.
\newblock (nearly) sample-optimal sparse fourier transform.
\newblock In Chandra Chekuri, editor, {\em 25th Annual {ACM-SIAM} Symposium on Discrete Algorithms ({SODA} 2014)}, pages 480--499. {SIAM}, 2014.
\newblock \href {https://doi.org/10.1137/1.9781611973402.36} {\path{doi:10.1137/1.9781611973402.36}}.

\bibitem{IndykR08}
Piotr Indyk and Milan Ruzic.
\newblock Near-optimal sparse recovery in the {L1} norm.
\newblock In {\em 49th Annual {IEEE} Symposium on Foundations of Computer Science ({FOCS} 2008)}, pages 199--207. {IEEE} Computer Society, 2008.
\newblock \href {https://doi.org/10.1109/FOCS.2008.82} {\path{doi:10.1109/FOCS.2008.82}}.

\bibitem{Janson18}
Svante Janson.
\newblock Tail bounds for sums of geometric and exponential variables.
\newblock {\em Statistics \& Probability Letters}, 135:1--6, 2018.

\bibitem{JaspanFL15}
Oren~N. Jaspan, Roman Fleysher, and Michael~L. Lipton.
\newblock Compressed sensing mri: a review of the clinical literature.
\newblock {\em The British Journal of Radiology}, 88(1056):20150487, 2015.
\newblock URL: \url{http://dx.doi.org/10.1259/bjr.20150487}, \href {https://doi.org/10.1259/bjr.20150487} {\path{doi:10.1259/bjr.20150487}}.

\bibitem{Kapralov16}
Michael Kapralov.
\newblock Sparse fourier transform in any constant dimension with nearly-optimal sample complexity in sublinear time.
\newblock In Daniel Wichs and Yishay Mansour, editors, {\em 48th Annual {ACM} Symposium on Theory of Computing ({STOC} 2016)}, pages 264--277. {ACM}, 2016.
\newblock \href {https://doi.org/10.1145/2897518.2897650} {\path{doi:10.1145/2897518.2897650}}.

\bibitem{Kapralov17}
Michael Kapralov.
\newblock Sample efficient estimation and recovery in sparse {FFT} via isolation on average.
\newblock In Chris Umans, editor, {\em 58th Annual {IEEE} Symposium on Foundations of Computer Science ({FOCS} 2017)}, pages 651--662. {IEEE} Computer Society, 2017.
\newblock \href {https://doi.org/10.1109/FOCS.2017.66} {\path{doi:10.1109/FOCS.2017.66}}.

\bibitem{KarpLS04}
Richard~M. Karp, Michael Luby, and Amin Shokrollahi.
\newblock Finite length analysis of {LT} codes.
\newblock In {\em {IEEE} International Symposium on Information Theory ({ISIT} 2004)}, page~39. {IEEE}, 2004.
\newblock \href {https://doi.org/10.1109/ISIT.2004.1365074} {\path{doi:10.1109/ISIT.2004.1365074}}.

\bibitem{LarsenNNT19}
Kasper~Green Larsen, Jelani Nelson, Huy~L. Nguyen, and Mikkel Thorup.
\newblock Heavy hitters via cluster-preserving clustering.
\newblock {\em Commun. {ACM}}, 62(8):95--100, 2019.
\newblock \href {https://doi.org/10.1145/3339185} {\path{doi:10.1145/3339185}}.

\bibitem{LiXW13a}
Shancang Li, Li~Da Xu, and Xinheng Wang.
\newblock Compressed sensing signal and data acquisition in wireless sensor networks and internet of things.
\newblock {\em {IEEE} Trans. Ind. Informatics}, 9(4):2177--2186, 2013.
\newblock \href {https://doi.org/10.1109/TII.2012.2189222} {\path{doi:10.1109/TII.2012.2189222}}.

\bibitem{LiYPPR19}
Xiao Li, Dong Yin, Sameer Pawar, Ramtin Pedarsani, and Kannan Ramchandran.
\newblock Sub-linear time support recovery for compressed sensing using sparse-graph codes.
\newblock {\em {IEEE} Trans. Inf. Theory}, 65(10):6580--6619, 2019.
\newblock \href {https://doi.org/10.1109/TIT.2019.2921757} {\path{doi:10.1109/TIT.2019.2921757}}.

\bibitem{LubyMSS01}
Michael Luby, Michael Mitzenmacher, Mohammad~Amin Shokrollahi, and Daniel~A. Spielman.
\newblock Efficient erasure correcting codes.
\newblock {\em {IEEE} Trans. Inf. Theory}, 47(2):569--584, 2001.
\newblock \href {https://doi.org/10.1109/18.910575} {\path{doi:10.1109/18.910575}}.

\bibitem{LustigDSP08}
Michael Lustig, David~L. Donoho, Juan~M. Santos, and John~M. Pauly.
\newblock Compressed sensing {MRI}.
\newblock {\em {IEEE} Signal Process. Mag.}, 25(2):72--82, 2008.
\newblock \href {https://doi.org/10.1109/MSP.2007.914728} {\path{doi:10.1109/MSP.2007.914728}}.

\bibitem{MintonP14}
Gregory~T. Minton and Eric Price.
\newblock Improved concentration bounds for count-sketch.
\newblock In Chandra Chekuri, editor, {\em 25th Annual {ACM-SIAM} Symposium on Discrete Algorithms ({SODA} 2014)}, pages 669--686. {SIAM}, 2014.
\newblock \href {https://doi.org/10.1137/1.9781611973402.51} {\path{doi:10.1137/1.9781611973402.51}}.

\bibitem{Mitzenmacher09}
Michael Mitzenmacher.
\newblock Some open questions related to cuckoo hashing.
\newblock In Amos Fiat and Peter Sanders, editors, {\em 17th Annual European Symposium on Algorithms ({ESA} 2009)}, volume 5757 of {\em Lecture Notes in Computer Science}, pages 1--10. Springer, 2009.
\newblock \href {https://doi.org/10.1007/978-3-642-04128-0\_1} {\path{doi:10.1007/978-3-642-04128-0\_1}}.

\bibitem{Molloy04}
Michael Molloy.
\newblock The pure literal rule threshold and cores in random hypergraphs.
\newblock In J.~Ian Munro, editor, {\em 15th Annual {ACM-SIAM} Symposium on Discrete Algorithms ({SODA} 2004)}, pages 672--681. {SIAM}, 2004.
\newblock URL: \url{http://dl.acm.org/citation.cfm?id=982792.982896}.

\bibitem{NakosS19}
Vasileios Nakos and Zhao Song.
\newblock Stronger {$\ell_2/\ell_2$} compressed sensing; without iterating.
\newblock In Moses Charikar and Edith Cohen, editors, {\em 51st Annual {ACM} Symposium on Theory of Computing ({STOC} 2019)}, pages 289--297. {ACM}, 2019.
\newblock \href {https://doi.org/10.1145/3313276.3316355} {\path{doi:10.1145/3313276.3316355}}.

\bibitem{PaghR01}
Rasmus Pagh and Flemming~Friche Rodler.
\newblock Cuckoo hashing.
\newblock In Friedhelm~Meyer auf~der Heide, editor, {\em 9th Annual European Symposium on Algorithms ({ESA} 2001)}, volume 2161 of {\em Lecture Notes in Computer Science}, pages 121--133. Springer, 2001.
\newblock \href {https://doi.org/10.1007/3-540-44676-1\_10} {\path{doi:10.1007/3-540-44676-1\_10}}.

\bibitem{Price11}
Eric Price.
\newblock Efficient sketches for the set query problem.
\newblock In Dana Randall, editor, {\em 22nd Annual {ACM-SIAM} Symposium on Discrete Algorithms ({SODA} 2011)}, pages 41--56. {SIAM}, 2011.
\newblock \href {https://doi.org/10.1137/1.9781611973082.4} {\path{doi:10.1137/1.9781611973082.4}}.

\bibitem{PriceW11}
Eric Price and David~P. Woodruff.
\newblock {$1+\epsilon$}-approximate sparse recovery.
\newblock In Rafail Ostrovsky, editor, {\em 52nd Annual {IEEE} Symposium on Foundations of Computer Science ({FOCS} 2011)}, pages 295--304. {IEEE} Computer Society, 2011.
\newblock \href {https://doi.org/10.1109/FOCS.2011.92} {\path{doi:10.1109/FOCS.2011.92}}.

\bibitem{Price13}
Eric~C. Price.
\newblock {\em Sparse recovery and Fourier sampling}.
\newblock PhD thesis, Massachusetts Institute of Technology, 2013.
\newblock URL: \url{https://hdl.handle.net/1721.1/85458}.

\bibitem{Spielman96}
Daniel~A. Spielman.
\newblock Linear-time encodable and decodable error-correcting codes.
\newblock {\em {IEEE} Trans. Inf. Theory}, 42(6):1723--1731, 1996.
\newblock \href {https://doi.org/10.1109/18.556668} {\path{doi:10.1109/18.556668}}.

\end{thebibliography}

\end{document}